\documentclass[11pt,a4paper]{article}
\usepackage[english]{babel} 
\usepackage{amsmath}
\usepackage{mathtools}
\usepackage{amsfonts}
\usepackage{amssymb}
\usepackage{amsthm}
\usepackage{makeidx}
\usepackage{graphicx}
\usepackage{natbib}
\usepackage{setspace}
\usepackage{float}
\usepackage{hyperref}
\usepackage[margin=1in]{geometry}
\usepackage{setspace}
\usepackage[titletoc]{appendix}
\hypersetup{
	colorlinks   = true,
	citecolor    = blue,
	urlcolor = blue
}

\author{Marcelo Brutti Righi\footnote{We are grateful for the financial support of CNPq (Brazilian Research Council) projects number 302369/2018-0 and 407556/2018-4.}
	\\
	\textit{Federal University of Rio Grande do Sul}\\
	marcelo.righi@ufrgs.br}

\title{Star-shaped acceptability indexes}
\date{}

\newtheorem{Def}{Definition}[section]
\newtheorem{Thm}[Def]{Theorem}
\newtheorem{Prp}[Def]{Proposition}

\newtheorem{Crl}[Def]{Corollary}
\theoremstyle{definition}
\newtheorem{Exm}[Def]{Example}

\theoremstyle{remark}
\newtheorem{Rmk}[Def]{Remark}

\numberwithin{equation}{section}

\begin{document}
	\maketitle
	\begin{abstract}
	We propose the star-shaped acceptability indexes as generalizations of both the approaches of \cite{Cherny2009} and \cite{Gianin2013} in the same vein as star-shaped risk measures generalize both the classes of coherent and convex risk measures in \cite{Castagnoli2021}. We characterize acceptability indexes through star-shaped risk measures and star-shaped acceptance sets as the minimum of a family of quasi-concave acceptability indexes. Further, we introduce concrete examples under our approach linked to Value at Risk, risk-adjusted reward on capital, reward-based gain-loss ratio, and monotone reward-deviation ratio.
\end{abstract}	

\section{Introduction}

The task of assessing the performance of financial investments is central. Since the beginning of modern financial theory, indexes such as the Sharpe ratio have been considered to assess the trade-off between risk and return. In the last decade, performance has been analyzed through acceptability indexes since the seminal paper of \cite{Cherny2009}. In that paper, the authors present an axiomatic theory with properties that a performance measure must fulfill. That is, better results are desired, diversification is beneficial, and the scale does not affect the decision. Further, they explore the connection to coherent risk measures and their acceptance sets in the sense of \cite{Artzner1999} and \cite{Delbaen2002}. The main feature is that a performance measure possesses the desired properties if and only if it can be represented by a family of acceptability indexes linked to coherent risk measures, representing some level of acceptability.

This axiomatic framework contemplates a wide number of specific examples that are used both in literature and industry. Many authors study specific measures. See \cite{Chen2011}, \cite{Schuhmacher2014} and \cite{Zhitlukhin2021} for the relation with  Sharpe ratio,  \cite{Biagini2013},  \cite{Zakamouline2014} and \cite{Voelzke2015} regarding Gain-loss ratio (GLR), \cite{Zakamouline2010} for the risk-adjusted return on capital (RAROC), and \cite{Mondal2020} for the so-called upside beta. Furthermore, extensions of the initial approach in \cite{Cherny2009} are conducted:  \cite{Bielecki2013}, \cite{Bielecki2014} and \cite{Bielecki2016} to a conditional dynamic framework; \cite{Kountzakis2020} for stochastic processes, and \cite{Zeng2019} for a multivariate context. Moreover, other advances are discussed, such as qualitative robustness in \cite{Rossello2015} and optimization in \cite{Kovacova2020}.

One of the properties demanded in these papers is scale invariance, which assures that acceptability does not change with position size. However, the size of a financial position can affect its performance.
This topic is relevant in financial market descriptions and their intrinsic dynamics. For risk measures theory, positive homogeneity plays this role. \cite{Follmer2002}, and \cite{Frittelli2002} argue against this property and the consequent sub-additivity assumptions adopted in the framework of coherent risk measures by introducing the class of convex risk measures. Further discussion on the relationship between risk measures and liquidity risk is done in \cite{Acerbi2008}, and \cite{Lacker2018}.

In order to incorporate this feature in the theory of performance measures and acceptability indexes, \cite{Gianin2013} relax the scale invariance property. Furthermore, the authors represent their framework's acceptability indexes under convex risk measures. This kind of convex acceptability family is also considered in \cite{Drapeau2013} and \cite{Frittelli2014}, where the key property for the functional is quasi-concavity. Such framework is extended to the conditional dynamic context in \cite{Biagini2014}.

Despite the role of liquidity, there is also debate on the benefits of diversification. More specifically, convexity, which is equivalent to sub-additivity under positive homogeneity, of risk measures can be misleading. See \cite{Dhaene2008} for instance. This issue could be present when risk is allocated competitively, when a non-concave utility function is involved in risk assessment, or when capital requirements are aggregated in robust ways that do not preserve convexity. In this sense, \cite{Castagnoli2021} propose the class of star-shaped risk measures, which can be understood as a generalization for positive homogeneity and convexity. The nomenclature comes from the star-shaped property of the generated acceptance set. This class allows for the most used non-convex risk measure, the Value at Risk (VaR), to be in the same class as the convex risk measures. \cite{Liebrich2021} explores allocations of star-shaped risk measures. In contrast, \cite{Moresco2021} relate them to the broader class of monetary risk measures, \cite{Herdegen2021} applies it to sensitivity to large losses and regulatory arbitrage, \cite{Laeven2023} relates to law invariance, \cite{Tian2023} explores a dynamic framework, while \cite{Nie2024} considers a set-valued setup.

Based on this discussion, there would be a gain in generality by considering indexes that are not restricted to the coherent or convex cases in order to be able to avoid both scale invariance, which is linked to liquidity issues, and quasi-concavity, which is connected to the role for diversification. The reasoning for star-shapedness as a sensible property
is that any scaled reduction is possible if a position is acceptable at some level. More specifically, the meaning of this property is evident: reducing the exposure to an acceptable position at some level cannot make it unacceptable at this same level. For instance, the change of quasi-concavity by star-shapedness is
motivated by the lack of empirical support for the general feasibility of mergers. Thus, star-shapedness penalizes
the concentration of risk and the ensuing liquidity problems. However, unlike quasi-concavity, star-shapedness does not take a position on the effects of merge-and-downsize strategies.

Under this context, our first objective and contribution is to extend the main representation results of acceptability indexes for this framework of star-shapedness, which is based on the maximum of a family of quasi-concave acceptability indexes and, equivalently, by an increasing family of star-shaped risk measures. Our second contribution is to propose concrete non-trivial classes of star-shaped acceptability indexes, studying their representation in the light of our main results. Such representations are crucial for practical implementation and are non-trivial from a technical standpoint. In practical matters, an agent's acceptability is based on some family of utilities or risk measures, and a regulator or decision maker considers some family of individual acceptability indexes. Both situations are covered in our paper. There is no paper dealing with this setup.

Section \ref{sec:app} exposes notations, definitions, and background material. In Section \ref{sec:results},  we propose the star-shaped acceptability indexes as generalizations of both the approaches of \cite{Cherny2009} and \cite{Gianin2013} in the same vein as star-shaped risk measures generalize both the classes of coherent and convex risk measures in \cite{Castagnoli2021}. We expose the main theoretical results, with an emphasis on Theorem \ref{Thm:main}, which characterize acceptability indexes through star-shaped risk measures, star-shaped acceptance sets, and as the minimum of some family of quasi-concave acceptability indexes. These results are the building block for the subsequent sections and examples in the paper.

In Section
\ref{sec:var} we use the fact that a direct way to produce acceptability indexes is to consider an increasing family of star-shaped risk measures/star-shaped acceptance sets to study the most used risk measure in literature, the VaR, which is not contemplated in the coherent and convex cases. From VaR, we also consider distortions in the Choquet integral format. Expected Shortfall (ES) is a prominent example, but there are also ones with non-concave distortions.
This kind of distortion-based acceptability
indexes are mainly used for pricing and hedging purposes, especially in the literature on conic finance. See the book of \cite{Madan2016} for a review. However, they are less used as tools for performance, mostly due to their lack of interpretation. Nonetheless, in our context, the intuition is for a risk-averse family expectation, which are risk measures that give more weight to losses and increase risk aversion. Furthermore, we use this setup to characterize law invariant or second-order stochastic dominance, preserving star-shaped acceptability indexes. This issue is connected to consistent risk measures as proposed by \cite{Mao2020}. 

In Section \ref{sec:ratio}, we explore acceptability indexes based on ratios of reward and risk measures. In subsection \ref{sec:raroc}, we focus on the widespread performance measure risk-adjusted return in capital (RAROC), which is a ratio between the return, measured as expectation, and a risk measure. We make two extensions/generalizations for this function. First, we allow the risk measure in the denominator to be star-shaped. Second, we extend the numerator from expectation to some reward measure, defined as the negative of some risk measure. In subsection \ref{sec:GLR} we explore another reward/risk criterion: a ratio acting as a performance measure. The gain-loss ratio (GLR), a ratio between the expectation of both positive and negative parts, that \cite{Bernardo2000} propose as an alternative to the Sharpe ratio, which can lead to good deals. We consider a related acceptability formulation for the GLR that is not quasi-concave but a star-shaped acceptability index. In subsection \ref{sec:Dev}, we consider another possibility for a ratio that is the one between reward and deviation measures, in the sense of \cite{Rockafellar2006}, \cite{Rockafellar2013} and \cite{Righi2018a}. Such a structure lacks monotonicity. \cite{Zhitlukhin2021} considers a monotone version of the Sharpe ratio, which we then extend this approach into our framework of star-shaped acceptability indexes by allowing a star-shaped reward in the numerator and a deviation in the denominator.

In section \ref{sec:ill}, we evaluate the acceptability indexes described above for a simple position in a spot market portfolio. The sample considered is composed of stocks from the S\&P100 index. Such illustration allows us to explore the
numerical magnitudes and the types of values one may expect in both absolute and relative matters. This illustration is also useful for empirical problems such as pricing under the physical measure since if an agent has access to historical data, he/she may compute the price that attains a particular level of
acceptability. We consider examples of star-shaped acceptability indexes mainly based on quantiles and their most usual quasi-concave or coherent counterparts. Results make it possible to verify that they present similar average values, indicating that they can be used for performance evaluation, producing the same evaluations but with more theoretical generality. Nonetheless, regarding the changes in distinct sub-samples, there are larger discrepancies for the quasi-concave/coherent performance measures concerning the star-shaped ones, reflecting some robustness for the latter class.

\section{{Background material}}\label{sec:app}

We consider a probability space $(\Omega,\mathcal{F},\mathbb{P})$. All equalities and inequalities are in the $\mathbb{P}-a.s.$ sense. Let $L^0=L^0(\Omega,\mathcal{F},\mathbb{P})$ and $L^{\infty}=L^{\infty}(\Omega,\mathcal{F},\mathbb{P})$ be the spaces of (equivalence classes under $\mathbb{P}-a.s.$ equality of) finite and essentially bounded random variables, respectively. When not explicit, we consider in  $L^\infty$ its norm topology. We define $1_A$ as the indicator function for an event $A\in\mathcal{F}$. We identify constant random variables with real numbers. 
We denote by $X_n\rightarrow X$ convergence in the $L^\infty$ essential supremum norm $\lVert \cdot\rVert_{\infty}$, whereas $\lim\limits_{n\rightarrow\infty}X_n=X$ indicates $\mathbb{P}-a.s.$ convergence. The notation $X\succeq Y$, for $X,Y\in L^\infty$, indicates second-order stochastic dominance, that is, $E[f(X)]\leq E[f(Y)]$ for any increasing convex function $f\colon\mathbb{R}\rightarrow\mathbb{R}$. 

Let $\mathcal{P}$ be the set of all probability measures on $(\Omega,\mathcal{F})$. We denote by $E_{\mathbb{Q}}[X]=\int_{\Omega}Xd\mathbb{Q}$, $F_{X, \mathbb{Q}}(x)=\mathbb{Q}(X\leq x)$, and $F_{X, \mathbb{Q}}^{-1}(\alpha)=\inf\left\lbrace x:F_{X, \mathbb{Q}}(x)\geq\alpha\right\rbrace $, the expected value, the (increasing and right-continuous) probability function, and its left quantile for $X\in L^\infty$ concerning $\mathbb{Q}\in\mathcal{P}$. We write $X\overset{\mathbb{Q}}\sim Y$ when $F_{X,\mathbb{Q}}=F_{Y,\mathbb{Q}}$. We drop subscripts indicating probability measures when $\mathbb{Q}=\mathbb{P}$. Furthermore, let $\mathcal{Q}\subset\mathcal{P}$ be the set of probability measures $\mathbb{Q}$ that are absolutely continuous with respect to $\mathbb{P}$, with Radon--Nikodym derivative $\frac{d\mathbb{Q}}{d\mathbb{P}}$. 

We begin by defining the main object of this paper, the acceptability indexes and their properties. These properties are detailed in \cite{Cherny2009}, and \cite{Gianin2013}.

\begin{Def}\label{def:accept}
	A functional $\alpha:L^\infty\rightarrow[0,\infty]$ is an acceptability index. It may have the following properties: 
	
	\begin{enumerate}
		\item Monotonicity: if $X \geq Y$, then $\alpha(X) \geq \alpha(Y),\:\forall\: X,Y\in L^\infty$.
		\item Quasi-concavity: $\alpha(\lambda X+(1-\lambda)Y)\geq \min\{\alpha(X),\alpha(Y)\},\:\forall\: X,Y\in L^\infty,\:\forall\:\lambda\in[0,1]$.
		\item Scale invariance: $\alpha(\lambda X)= \alpha(X),\:\forall\: X\in L^\infty,\:\forall\:\lambda > 0$.
		\item Star-shapedness: $\alpha(0)=\infty$ and $\alpha(\lambda X)\leq \alpha(X),\:\forall\: X\in L^\infty,\:\forall\:\lambda \geq 1$.
		\item Fatou continuity: if $\lim\limits_{n\rightarrow\infty}X_n=X$ and $\alpha(X_n)\geq x$, then $\alpha(X) \geq x$ for any $x\geq0$, $\forall\:\{X_n\}_{n=1}^\infty$ bounded in $L^\infty$ norm and for any $X\in L^\infty$.
		\item  Weak expectation consistency:  $\alpha(C)=0,\:\forall\:C\in \mathbb{R}$ with $C<0$.
		\item Expectation consistency: if $E[X]<0$, then $\alpha(X)=0$ and if $E[X]>0$, then $\alpha(X)>0,\:\forall\:X\in L^\infty$.
		\item Arbitrage consistency: if $X\geq0$ and $\mathbb{P}(X>0)>0$, then $\alpha(X)=\infty,\:\forall\: X\in L^\infty$.
		\item SSD consistency: if $X\succeq Y$, then  $\alpha(X) \geq \alpha(Y),\:\forall\: X,Y\in L^\infty$.
		\item Law invariance: if $F_X=F_Y$, then $\alpha(X)=\alpha(Y),\:\forall\:X,Y\in L^\infty$.	
	\end{enumerate}
	An acceptability index is called coherent if it satisfies (i), (ii), (iii), (v), and (vi); quasi-concave if it satisfies (i), (ii), (iv), (v) and (vi); star-shaped if it fulfills (i), (iv), (v) and (vi);  weak expectation, expectation, arbitrage, or SSD consistent if it satisfies, respectively (vi), (vii), (viii) and (ix); and law invariant if it satisfies (x).
\end{Def}

\begin{Rmk}
	In our star-shaped approach, and also on that of \cite{Gianin2013}, $0$ is acceptable at any level, i.e., $\alpha(0)=\infty$. Intuitively, this means that a null payoff, which is morally the same as doing nothing, is always acceptable. This acceptability is natural in the framework of monetary risk measures since, in that case, under normalization ($\rho(0)=0$), the null position is always acceptable. This property is important since it implies, under Monotonicity, that if $X\geq 0$, then $\alpha(X)\geq\alpha(0)=\infty$. In particular, Arbitrage consistency is satisfied. This definition is distinct from the approach in \cite{Cherny2009}, where most examples fulfill $\alpha(0)=0$, which under Monotonicity is stronger than Weak expectation consistency. In order to unify both approaches, it is necessary to make mild adjustments, for instance, to assume Weak expectation consistency, as is the case for our acceptance indexes, which causes no harm. Furthermore, Proposition \ref{prp:star} shows that under both quasi-concavity or scale invariance, presents in the approach of \cite{Cherny2009}, we have star shapedness, i.e., in such contexts, it is sufficient to have $\alpha(0)=\infty$.
\end{Rmk}

A key aspect of acceptability indexes is their representation under risk measures. In this sense, we now expose the definition of this concept. We refer to \cite{Delbaen2012} and \cite{Follmer2016} for more details regarding these properties.

\begin{Def}\label{def:risk}
	A risk measure is a functional $\rho\colon L^\infty\rightarrow\mathbb{R}$. It may have the following properties: 
	
	\begin{enumerate}
		\item (Anti-)Monotonicity: If $X \leq Y$, then $\rho(X) \geq \rho(Y),\:\forall\: X,Y\in L^\infty$.
		\item Translation invariance: $\rho(X+C)=\rho(X)-C,\:\forall\: X\in L^\infty,\:\forall\:C \in \mathbb{R}$.
		\item Convexity: $\rho(\lambda X+(1-\lambda)Y)\leq \lambda \rho(X)+(1-\lambda)\rho(Y),\:\forall\: X,Y\in L^\infty,\:\forall\:\lambda\in[0,1]$.
		\item Positive homogeneity: $\rho(\lambda X)=\lambda \rho(X),\:\forall\: X\in L^\infty,\:\forall\:\lambda \geq 0$.
		\item Star-shapedness: $\rho(\lambda X)\geq \lambda\rho(X),\:\forall\: X\in L^\infty,\:\forall\:\lambda \geq 1$.
		\item Law invariance: If $F_X=F_Y$, then $\rho(X)=\rho(Y),\:\forall\:X,Y\in L^\infty$.
		\item Fatou continuity:  If $\lim\limits_{n\rightarrow\infty}X_n=X$, then $\rho(X) \leq \liminf\limits_{n\rightarrow\infty} \rho( X_{n})$, $\forall\:\{X_n\}_{n=1}^\infty$ bounded in $L^\infty$ norm and for any $X\in L^\infty$.
	\end{enumerate}
	A risk measure is called monetary if it satisfies (i) and (ii);  convex if it is monetary and satisfies (iii); coherent if it is convex and satisfies (iv); star-shaped if it is monetary and satisfies (v);  law invariant if it satisfies (vi), and Fatou continuous if it attends (vi). Unless
	otherwise stated, we assume that risk measures are normalized in the sense that $\rho(0) = 0$. Its
	acceptance set is defined as $\mathcal{A}_\rho = \{X \in L^\infty\colon \rho(X) \leq 0\}$.
\end{Def}

A fundamental concept for both risk measures and acceptability indexes is the acceptance set. We now define this concept and expose some properties. Propositions 2.1 and 2.2 of \citet{Artzner1999}, Propositions 4.6 and 4.7 of \citet{Follmer2016}, and Proposition 2 of \cite{Castagnoli2021} expose a direct interplay between the properties of risk measures and acceptance sets.

\begin{Def}
	An acceptance set is a subset $\mathcal{A} \subseteq L^\infty$. It may fulfill the following properties:
	\begin{enumerate}
		\item Monotonicity: if $X \in \mathcal{A}$ and $X\leq Y$, then  $Y \in  \mathcal{A},\:\forall\:X,Y\in L^\infty$.
		\item Monetarity: $\inf \{ m \in \mathbb{R}\colon m \in \mathcal{A} \}=0$.
		\item [(iii)] Convexity:  if $X$ and $Y \in  \mathcal{A}$, then  $\lambda X +(1-\lambda) Y \in  \mathcal{A},\:\forall\:X,Y\in L^\infty,\:\forall\:\lambda\in[0,1]$. 
		\item [(iv)] Conicity: if $X \in  \mathcal{A}$, then implies $\lambda X \in  \mathcal{A},\:\forall\:X\in L^\infty,\:\forall\:\lambda\geq0$.
		\item [(v)] Star-shapedness: if $X \in  \mathcal{A}$, then $\lambda X \in  \mathcal{A},\:\forall\:X\in L^\infty,\:\forall\:\lambda\in[0,1]$.
	\end{enumerate}
	An acceptance set is called monetary if it satisfies (i) and (ii);  convex if it is monetary and satisfies (iii); coherent if it is convex and satisfies (iv); star-shaped if it is monetary and satisfies (v). Its induced risk measure is given by $\rho_\mathcal{A}(X) = \inf \{ m \in \mathbb{R}\colon X + m \in \mathcal{A}\},\:\forall\:X\in L^\infty$.
\end{Def}

We have the following main representations in the literature for coherent and quasi-concave acceptability indexes.

\begin{Thm}[Theorem 1 in \cite{Cherny2009}, Proposition 3 in \cite{Gianin2013}]\label{Thm:back}
	A functional $\alpha\colon L^\infty\to[0,\infty]$  is a quasi-concave (resp. coherent) acceptability index if and only if there is an increasing family $\{\rho_{x}\}_{\{x>0\}}$ of Fatou continuous convex (resp. coherent) risk measures such that the representation holds
	\begin{equation}\label{eq:dualbc}
		\alpha(X)=\sup\left\lbrace  x>0\colon\rho_{x}(X)\leq 0\right\rbrace=\sup\left\lbrace  x>0\colon X\in\mathcal{A}_{\rho_x}\right\rbrace,\:\forall\:X\in L^\infty,
	\end{equation}
	with  the convention $\sup\emptyset=0$.  
\end{Thm}

\section{Main Theory}\label{sec:results}

We begin by exposing preliminary results regarding basic properties fulfilled by star-shaped acceptability indexes. We first show alternative equivalent formulations for star-shapedness.

\begin{Prp}\label{Lmm:star}
	Let $\alpha\colon L^\infty\to[0,\infty]$ with $\alpha(0)=\infty$. The following are equivalent:
	\begin{enumerate}
		\item $\alpha(\lambda X)\leq\alpha(X),\:\forall\: X\in L^\infty,\:\forall\:\lambda\geq 1$.
		\item $\alpha(\lambda X)\geq\alpha(X),\:\forall\: X\in L^\infty,\:\forall\:\lambda\in[0,1]$.
		\item $\lambda\to\alpha(\lambda X)$ is decreasing $\forall\: X\in L^\infty$.
	\end{enumerate}
	
\end{Prp}

\begin{proof}
	We have for any $X\in L^\infty$ the following chain of implications:
	
	(i)$\implies$ (ii). Let $\lambda\in[0,1]$. If $\lambda=0$ the result is trivial. If $\lambda>0$, then $\frac{1}{\lambda}\geq 1$. Hence, we have $\alpha(X)=\alpha\left(\lambda \frac{1}{\lambda}X\right) \leq \alpha\left(\lambda X\right) $.
	
	(ii)$\implies$(iii). Let $\lambda_1\geq\lambda_2\geq 0$. If $\lambda_1=0$ the result is trivial. If $\lambda_1>0$, then $0\leq\frac{\lambda_2}{\lambda1}\leq 1$. We then have that $\alpha(\lambda_2 X)=\alpha\left(\frac{\lambda_2}{\lambda_1}\lambda_1 X\right)\geq\alpha(\lambda_1 X)$.
	
	(iii)$\implies$(i). Let $\lambda\geq 1$. Then, we directly have that $\alpha(X)=\alpha(1X)\geq \alpha(\lambda X)$.
\end{proof}

We now prove an interplay of Star-shapedness to Scale invariance and Quasi-concavity. In particular, coherence implies quasi-concavity, which in turn implies star-shapedness. Further, under quasi-supperaditivity, i.e. $\alpha(X+Y)\geq\alpha(X)\wedge\alpha(Y)$ for any $X,Y\in L^\infty$, the three classes coincide. 

\begin{Prp}\label{prp:star}
	Let $\alpha\colon L^\infty\to[0,\infty]$ with $\alpha(0)=\infty$. Then, both Quasi-concavity and Scale invariance implies Star-shapedness. The converse is true if $\alpha$ is quasi-superadditive.
\end{Prp}

\begin{proof}
	Let  $\lambda\in[0,1]$ and $X,Y\in L^\infty$. If $\alpha$ is quasi-concave, then \[\alpha(\lambda X)=\alpha(\lambda X+(1-\lambda)0)\geq \alpha(X)\wedge\alpha(0)=\alpha(X),\] which, according to Proposition \ref{Lmm:star}, is star-shapedness. We have $\alpha(\lambda X)=\alpha(X)$ for Scale invariance. For the converse, we have \[\alpha(\lambda X+(1-\lambda)Y)\geq \alpha(\lambda X)\wedge\alpha((1-\lambda)Y)\geq\alpha(X)\wedge\alpha(Y),\] which is Quasi-concavity. For Scale invariance, let $k\in\mathbb{N}$. Then, $\alpha(kX)\geq\alpha(X)\geq\alpha(kX)$. Thus, $\lambda\to\alpha(\lambda X)$ is constant and equal to $\alpha(X)$ in $\mathbb{N}$, which jointly to Proposition \ref{Lmm:star} implies it is constant in $\mathbb{R}_+$. 
\end{proof}

In the context of star-shaped acceptability indexes, we have the following Theorem. It generalizes the result in Theorem \ref{Thm:back} in a sense; it provides an interplay through representation among star-shaped acceptability indexes, star-shaped risk measures, star-shaped acceptance sets, and quasi-concave acceptability indexes.

\begin{Thm}\label{Thm:main}
	Let $\alpha\colon L^\infty\to[0,\infty]$ and assume  the convention $\sup\emptyset=0$.  Then the following claims are equivalent:
	\begin{enumerate}
		\item $\alpha$ is a star-shaped acceptability index.
		\item there is a decreasing family of star-shaped acceptance sets $\{\mathcal{A}_x\}_{x>0}$ closed in the $a.s.$ convergence of bounded sequences such that the representation holds
		\begin{equation}\label{eq:dual3}
			\alpha(X)=\sup\left\lbrace  x>0\colon X\in\mathcal{A}_x\right\rbrace,\:\forall\:X\in L^\infty.
		\end{equation}
		\item there is an increasing family $\{\rho_{x}\}$ of Fatou continuous star-shaped risk measures such that the representation holds
		\begin{equation}\label{eq:dual}
			\alpha(X)=\sup\left\lbrace  x>0\colon\rho_{x}(X)\leq 0\right\rbrace,\:\forall\:X\in L^\infty.
		\end{equation}
	\end{enumerate}
	Moreover, each of these equivalent conditions implies the following
	\begin{enumerate}
		\item [(iv)]there is a family $\{\alpha_i\}_{i\in\mathcal{I}}$ of quasi-concave acceptability indexes such that the representation holds \begin{equation}\label{eq:alpha_max}
			\alpha(X)=\max_{i\in\mathcal{I}}\alpha_i(X),\:\forall\:X\in L^\infty.
		\end{equation}
		Such family can be chosen as the quasi-concave acceptability indexes dominated by $\alpha$, i.e., $\mathcal{I}=\{\beta\colon L^\infty\to[0,\infty]\colon\beta\:\text{is quasi-concave acceptability index and}\:\beta\leq\alpha\}$.
	\end{enumerate}
\end{Thm}

\begin{proof}
	(i)$\implies$ (ii).	Define for any $x>0$ the set \[\mathcal{A}_x=\{X\in L^\infty\colon\alpha(X)\geq x\}.\]  It is clear that $\mathcal{A}_x\subseteq\mathcal{A}_y$ for any $x>y$. Moreover, from the Monotonicity of $\alpha$, we have that each $\mathcal{A}_x$ is monotone. Furthermore, due to Star-shapedness of $\alpha$, $\mathcal{A}_x$ is star-shaped. Finally, from Weak expetation consistency, we get that $\inf\{m\in\mathbb{R}\colon m\in\mathcal{A}_x\}=0>-\infty$. For closedness, let $\{X_n\}\subseteq \mathcal{A}_x$  be bounded and $\lim\limits_{n\to\infty}X_n=X$. Then, $\alpha(X_n)\geq x$. From Fatou continuity of $\alpha$, we obtain that $\alpha(X)\geq x$, which implies $X\in\mathcal{A}_x$.
	
	(ii)$\implies$ (iii). Define for any $x>0$ the map \[\rho_x(X)=\rho_{\mathcal{A}_x}(X)=\inf\{m\in\mathbb{R}\colon X+m\in \mathcal{A}_x\},\:\forall\:X\in L^\infty.\] We now show that $\{\rho_x\}_{x>0}$ is an increasing family of Fatou continuous star-shaped risk measures. Since $\{\mathcal{A}_x\}_{x>0}$ is decreasing, we have that $\{\rho_x\}_{x>0}$ is increasing. Proposition 2 in \cite{Castagnoli2021} implies that each $\rho_x$ is a star-shaped risk measure and $\mathcal{A}_x=\mathcal{A}_{\rho_x}=\{X\in L^\infty\colon \rho_x(X)\leq 0\}$.	Hence, $\alpha(X)\geq x$ if and only if $\rho_x(X)\leq0$.
	For Fatou continuity of $\rho_x$ let $\{X_n\}$ be bounded such that $\lim\limits_{n\to\infty}X_n=X$ and $\rho_x(X_n)\leq m$ for any $n\in\mathbb{N}$. By Translation invariance of $\rho_x$, we get $\rho_x(X_n+m)\leq 0$. Thus $\{X_n+m\}\subseteq \mathcal{A}_x$. By closedness in the $a.s.$ convergence of bounded sequences, we obtain $X+m\in \mathcal{A}_x$. This is equivalent to $\rho_x(X)\leq m$, which implies $\rho(X)\leq \liminf\limits_{n\to\infty}\rho_x(X_n)$.
	
	(iii)$\implies$(i). Let $\alpha$ be given by		\eqref{eq:dual}. Monotonicity follows since for any $X\geq Y$, we have $\rho_x(X)\leq \rho_x(Y)$ for any $x>0$. For Star-shapedness, let $\lambda>1$. Then $\rho_x(\lambda X)\geq \lambda\rho_x(X)$ for any $\rho_x$. This fact implies		\[\alpha(\lambda X)=\sup\left\lbrace  x>0\colon\rho_{x}(\lambda X)\leq 0\right\rbrace\leq \sup\left\lbrace  x>0\colon\rho_{x}(X)\leq 0\right\rbrace=\alpha(X).\] Further, since $\rho_x(0)=0$ for any $x$ we directly have \[\alpha(0)=\sup\left\lbrace  x>0\colon\rho_{x}(0)\leq 0\right\rbrace=\sup\{x>0\}=\infty.\]	
	For Weak expectation consistency, let $C\in\mathbb{R}$ such that  $C<0$. By Monotonicity, we get $0\leq\alpha (C)=\sup\left\lbrace  x>0\colon C\geq 0\right\rbrace=0$. Regarding to Fatou continuity, let $\{X_n\}$ be bounded such that $\lim\limits_{n\to\infty}X_n= X$ and $\alpha (X_n)\geq x$. This is equivalent to $\rho_x(X_n)\leq0$ for any $n\in\mathbb{N}$. By Fatou continuity of the family $\{\rho_x\}$ we get $\rho_x(X)\leq\liminf\limits_{n\to\infty}\rho_x(X_n)\leq 0$. Then, we obtain $\alpha(X)\geq x$. This is equivalent to $\alpha(X)\geq\limsup\limits_{n\to\infty}\alpha(X_n)$. 
	
	(iii)$\implies$(iv). Let $\alpha\colon L^\infty\to[0,\infty]$ be represented under the monotone star-shaped families $\{\rho_x\}_{x>0}$. For any $x>0$ and any $Y\in L^\infty$, we define  \[\mathcal{A}_{Y,x}=cl^{*}(conv(\{Y+\rho_x(Y)\}\cup\{0\})+L^\infty_+),\]  where $cl^{*}$ means the closure regarding weak$^{*}$ topology. Furthermore, note that if $X\in \mathcal{A}_{Y,x}$, then $X=\lim\limits_{n\to\infty}\lambda_n(Y+\rho_x(Y)+Z_n)$ for some $\{Z_n\}\subseteq L^\infty_+$ bounded and $\{\lambda_n\}\subseteq[0,1]$. This fact implies that \[\rho_x(X)\leq\liminf\limits_{n\to\infty}\lambda_n\rho_x(Y+\rho_x(Y)+Z_n)\leq\liminf\limits_{n\to\infty}\lambda_n(\rho_x(Y)-\rho_x(Y))=0.\] Thus, $\mathcal{A}_{Y,x}\subseteq\mathcal{A} _{\rho_x}$. We then have that $\mathcal{A}_{Y,x}$ is convex and weak$^{*}$ closed since $L^\infty_+$ is a convex cone, monotone since $\mathcal{A}_{Y,x}+L^\infty=\mathcal{A}_{Y,x}$, and fulfills the condition $\inf\{m\in\mathbb{R}\colon m\in\mathcal{A}_{Y,x}\}=0$. Then, $\rho_{\mathcal{A}_{Y,x}}$ is a Fatou continuous convex risk measure, see Theorem 4.33 in \cite{Follmer2016}, such that $\rho_{\mathcal{A}_{Y,x}}\geq\rho_x$ for any $Y\in L^\infty$ and any $x>0$. Nonetheless, since $X+\rho_x(X)\in\mathcal{A}_{X,x}$ we get \[\rho_{\mathcal{A}_{X,x}}(X)=\inf\{m\in\mathbb{R}\colon X+m\in\mathcal{A}_{X,x}\}\leq \rho_x(X).\] Hence, we obtain that $\rho_{\mathcal{A}_{X,x}}(X)=\rho_x(X)$, which immediately implies \[\rho_x(X)=\min\limits_{Y\in L^\infty}\rho_{\mathcal{A}_{Y,x}}(X),\:\forall\:X\in L^\infty,\:\forall\:x>0.\]
	Now, for any $Y\in L^\infty$ we define \[\alpha_Y(X)=\sup\{x>0\colon\rho_{\mathcal{A}_{Y,x}}(X)\leq 0\}.\] By Theorem \ref{Thm:back}, we have that each $\alpha_Y$ defines a quasi-concave acceptability index. Furthermore, for any $X\in L^\infty$ we have that \[\alpha_Y(X)\leq\sup\{x>0\colon\rho_{x}(X)\leq 0\}=\alpha(X)=\sup\{x>0\colon\rho_{\mathcal{A}_{X,x}}(X)\leq 0\}=\alpha_X(X).\]
	Hence, we have that \[\alpha(X)=\max\limits_{Y\in L^\infty}\alpha_Y(X),\:\forall\:X\in L^\infty.\] Moreover, let $\mathcal{I}=\{\beta\colon L^\infty\to[0,\infty]\colon\beta\:\text{is quasi-concave acceptability index and}\:\beta\leq\alpha\}$. We have that $\alpha(X)\geq\sup_{\mathcal{I}}\beta (X)$. Since $\alpha_X\in\mathcal{I}$, we have that \[\alpha(X)=\max\limits_{\mathcal{I}}\beta(X),\:\forall\:X\in L^\infty.\]
\end{proof}

\begin{Rmk}\label{Rmk:set}
	The set $\mathcal{I}$ in the representation \eqref{eq:alpha_max} is not unique. Nonetheless, under relaxation, we have some unique results. For any set $\mathcal{I}$ of quasi-concave acceptability indexes, define its relaxation as \[\mathcal{I}^{*}=\left\lbrace \beta\colon L^\infty\to[0,\infty]\colon\beta\:\text{is quasi-concave acceptability index and}\:\beta\leq\max\limits_{\mathcal{I}}\beta_i\right\rbrace .\] Thus, if $\alpha=\max\limits_{\mathcal{I}_1}\beta_i=\max\limits_{\mathcal{I}_2}\beta_i$, then we directly have that \[\mathcal{I}^{*}_1=\mathcal{I}_2^{*}=\{\beta\colon L^\infty\to[0,\infty]\colon\beta\:\text{is quasi-concave acceptability index and}\:\beta\leq\alpha\}.\] 
	In particular, these sets are dependent on $\alpha$ and not on a specific representation. Moreover, we can recover both $\mathcal{A}_x$ and $\rho_x$ from such representation since we have that
	\begin{align*} \mathcal{A}_x&=\left\lbrace X\in L^\infty\colon  \exists\:i\in\mathcal{I}\;\text{s.t.}\:\alpha_i(X)\geq x\right\rbrace=\bigcup_{\mathcal{I}}\{X\in L^\infty\colon\alpha_i(X)\geq x\}= \bigcup_{\mathcal{I}}\mathcal{A}^i_x.\\
		\rho_x(X)&=\inf\left\lbrace m\in\mathbb{R}\colon X+m\in\bigcup_{\mathcal{I}}\mathcal{A}^i_x\right\rbrace\\
		&=\inf\left\lbrace m\in\mathbb{R}\colon\exists\:i\in\mathcal{I}\;\text{s.t.} \rho_x^i(X)\leq m\right\rbrace=\inf\left(\inf\limits_\mathcal{I}\rho_x^i(X),\infty \right)=\inf\limits_\mathcal{I}\rho_x^i(X).
	\end{align*} 
\end{Rmk}

The converse implication of (iv)$\implies$(iii) in Theorem \ref{Thm:main} is not assured as Fatou continuity (in fact, any upper semicontinuity) of a family of maps is not in general preserved by the supremum operation. Nonetheless, by dropping Fatou continuity, we have an interesting
characterization result of a monotone, star-shaped, and weak expectation consistent functional.

\begin{Prp}
	Let $\alpha\colon L^\infty\to[0,\infty]$ and assume  the convention $\sup\emptyset=0$. Then, $\alpha$ is monotone, star-shaped, and weak expectation consistent if and only if there is a family $\{\alpha_i\}_{i\in\mathcal{I}}$ of monotone, quasi-concave and weak expectation consistent acceptability indexes such that the representation holds \begin{equation*}\label{eq:alpha_max2}
		\alpha(X)=\max_{i\in\mathcal{I}}\alpha_i(X),\:\forall\:X\in L^\infty.
	\end{equation*}
	Such family can be chosen as the quasi-concave acceptability indexes dominated by $\alpha$, i.e. $\mathcal{I}=\{\beta\colon L^\infty\to[0,\infty]\colon\beta\;\text{is quasi-concave and monotone acceptability index and}\;\beta\leq\alpha\}$. 
	
\end{Prp}

\begin{proof}
	The claim follows as  (iii)$\implies$(iv) in Theorem \ref{Thm:main} by considering the acceptance sets  $\mathcal{A}_{Y,x}=conv(\{Y+\rho_x(Y)\}\cup\{0\})+L^\infty_+$. The converse implication is straightforward.	
\end{proof}

\begin{Rmk}
	In the previous proposition, weak expectation consistency is important to achieve the result.
	This claim can be directly compared with Theorem 5 in \cite{Castagnoli2021} where translation
	invariance is assumed, and Theorem 5.1 of \cite{Han2022} where one should rely on quasistar-shapedness (see Remark \ref{rmk:lambda}) instead of star-shapedness when weak expectation consistency is replaced by
	normalization.
\end{Rmk}

In the conditions of Theorem \ref{Thm:main}, we can recover the original results of Theorem 1 in \cite{Cherny2009} and Proposition 3 in \cite{Gianin2013}. We have a Corollary regarding this topic, which has been stated without proof.

\begin{Crl}	
	\label{rmk:PH}Let $\alpha\colon L^\infty\to[0,\infty]$ and assume  the convention $\sup\emptyset=0$. If the equivalent properties in   Theorem \ref{Thm:main} hold, then we have the following: \begin{enumerate}
		\item  $\alpha$ has Scale invariance if and only if $\{\mathcal{A}_x\}_{x>0}$ is composed by cones if and only if each member in $\{\rho_x\}_{x>0}$ fulfills Positive Homogeneity. 
		\item $\alpha$ has Quasi-concavity if and only if $\{\mathcal{A}_x\}_{x>0}$ is composed by convex sets if and only if each member in $\{\rho_x\}_{x>0}$ fulfills Convexity.
	\end{enumerate}
\end{Crl}

We end this section by exposing some practical examples where our main results are directly applicable. The main focus is to exemplify where both the maximum acceptability indexes and the representation over star-shaped risk measures/utilities are useful.

\begin{Exm}
	Consider a financial institution with $K\in\mathbb{N}$ business lines (or even a regulation system with $K$ institutions), and the manager of each business line adopts quasi-concave acceptability indexes $\alpha_i(X),\:i=1,\dots,K$ to assess the performance from a
	given global position $X$. It is interesting to resume this $K$-dimensional information into one number to determine whether $X$ is deemed acceptable for such an institution. The
	most conservative choice can be obtained by $\alpha_{\min}(X)=\min_{i\in\{1,\dots,K\}}\alpha_i(X)$, in a way that the most pessimistic acceptability is taken in order to satisfy all business lines. It is easy to verify that such a function is also quasi-concave. Despite this solution being cautious, it can lead to a pathological situation where more
	rigorous business lines may be concerned that this approach ignores what the less cautious ones classify as acceptability. It would generate a moral hazard problem by incentivizing less rigorous business lines to support risky asset purchases with high acceptability. The funds cover the possible losses on the more cautious business lines. In order to circumvent this situation, the institution could decide to delegate the decision to the business line with the larger acceptability index while assuring only limited coverage for losses. Thus, the corresponding acceptability index would be $\alpha_{\max}(X)=\max_{i\in\{1,\dots,K\}}\alpha_i(X)$, which is, by Theorem \ref{Thm:main}, star-shaped but not quasi-concave.
\end{Exm}

\begin{Exm}
	A context where star-shaped acceptability indexes are necessary is non-concave utilities. Let $\{u_x\}_{x>0}$ represent concave utility functions with aversion parameter $x$. The interpretation is that $u_x$ represents more risk-averse attitudes/preferences as $x$ increases. In this sense, one can build for some position $X$ an acceptability index based on expected utility criteria over some threshold as $\alpha(X)=\sup\left\lbrace  x>0\colon E[u_x(X)]\geq K\right\rbrace$. This map is quasi-concave. Concavity of $u_x$ corresponds to a strong form of risk aversion, while more flexible utility functions with local convexities have been well documented, see  \cite{Landsberger1990} for instance. In this case, we consider the broader class of utility functions such
	that $\lambda \to \frac{u_x(\lambda y)}{\lambda}$ decreasing on $(0,\infty)$ for any $y\in\mathbb{R}$, which is equivalent to star-shapedness of $u_x$. This condition, which is weaker than concavity, implies that the acceptability index defined as $\alpha(X)=\sup\left\lbrace  x>0\colon E[u_x(X)]\geq K\right\rbrace$ is, by Theorem \ref{Thm:main}, star-shaped but not quasi-concave.
\end{Exm}

\begin{Exm}
	This example follows similar reasoning as that exposed in \cite{Castagnoli2021}. Assume that a supervising agency consists
	of a board of experts $i\in \mathcal{I}$, which proposes a quasi-concave acceptability index $\alpha_i$. Suppose the agency has to aggregate these opinions on some global portfolio $X$ for an institution. In that case, it will have $\alpha_f(X)=f(\alpha_\mathcal{I}(X))$, where $\alpha_\mathcal{I}(X)=\{\alpha_i(X),\:i\in\mathcal{I}\}$ and $f$ is some aggregation map, see \cite{Righi2019} for details on combination of risk measures. When $f$ is any weighting average under some weight scheme $\mu$, one gets, under the necessary measurability properties, $\alpha_f(X)=\int_\mathcal{I}\alpha_i(X)d\mu$. This map is quasi-concave, as is the $\alpha_{\min}$ from previous paragraphs. However, many other choices for $f$, such as median, general order statistics, and supremum, are typically not quasi-concave but lead to star-shaped acceptability indexes. In the following section \ref{sec:var}, we expose and study some concrete cases of acceptability indexes based on quantiles.
\end{Exm}

\section{Value at Risk based acceptability indexes}\label{sec:var}

Let VaR be defined as $VaR^p(X)=-F_{X}^{-1}(p),\:p\in[0,1]$, with $\mathcal{A}_{VaR^p}=\left\lbrace X\in L^\infty:\mathbb{P}(X<0)\leq p\right\rbrace $. We then consider $\rho_x=VaR^\frac{1}{1+x}$, which is increasing in $x$, as a building block. 
It is also possible to consider more general formulations such as $\rho_x=VaR^{u(x)}$ for a decreasing function $u\colon \mathbb{R}_+\to [0,1]$. Nonetheless, this paper focuses on this specific choice for the ratio since it is the simpler decreasing transformation $u$ from non-negative numbers to the unit interval. Moreover, it is easily handled, and such a formulation naturally appears in other contexts.

\begin{Def}
	The VaR based acceptability index is a functional $\alpha_{VaR}\colon L^\infty\to[0,\infty]$ defined as\begin{equation}
		\alpha_{VaR}(X)=\sup\left\lbrace x>0\colon VaR^{\frac{1}{1+x}}(X)\leq 0\right\rbrace. 
	\end{equation}
\end{Def}

Since VaR is a Fatou continuous monetary risk measure that fulfills positive homogeneity, it is star-shaped and, by Theorem \ref{Thm:main}, $\alpha_{VaR}$ is a star-shaped acceptability index. By Remark \ref{rmk:PH} and Proposition \ref{crl:LI}, we also have that this index is Scale and Law Invariant. Furthermore, we get that the acceptance sets in the representation \eqref{eq:dual} are as \[\mathcal{A}_x=\mathcal{A}_{VaR^{\frac{1}{1+x}}}=\left\lbrace X\in L^\infty\colon VaR^{\frac{1}{1+x}}(X)\leq 0\right\rbrace =\left\lbrace X\in L^\infty\colon \mathbb{P}(X<0)\leq\frac{1}{1+x}\right\rbrace.\]

Combinations of VaR at distinct levels once such combination function fulfills some properties; see \cite{Righi2019} for details. A typical formulation is the one for distortion risk measures, which are defined through Choquet integrals as \begin{align*}
	\rho_g(X)
	&=\int_{-\infty}^{0}(g(\mathbb{P}(-X\geq x)-1)dx+\int_{0}^{\infty}g(\mathbb{P}(-X\geq x))dx,
\end{align*}
where $g\colon[0,1]\to[0,1]$ is a distortion, which is normalized and increasing. Such measures are law-invariant, monetary, and positively homogeneous. In atomless probability spaces, the convexity of $\rho$ is equivalent to $g$ being concave. Furthermore, if $g_1\geq g_2$, then $\rho_{g_1}\geq\rho_{g_2}$. Thus, we have a star-shaped acceptability index by taking an increasing family of distortions, not necessarily concave.

\begin{Def}
	The distortion based acceptability index is a functional $\alpha_{G}\colon L^\infty\to[0,\infty]$ defined as\begin{equation}
		\alpha_{G}(X)=\sup\left\lbrace x>0\colon \rho_{g_x}(X)\leq 0\right\rbrace, 
	\end{equation}
	where $G=\{g_x\}_{x>0}$ is an increasing family of distortion functions.
\end{Def}

Possible choices for $G$ are $g_x(y)=1_{\{y\geq\frac{1}{1+x}\}}$, in which case we recover $\rho_x=VaR^\frac{1}{1+x}$. A concave choice is regarding the well-known Expected Shortfall (ES), which is a Fatou continuous law invariant coherent risk measure defined as $ES^{p}(X)=\frac{1}{p}\int_0^p VaR^s(X)ds,\:p\in(0,1]$ and $ES^0(X)=VaR^0(X)=-\operatorname{ess}\inf X$. We have $\mathcal{A}_{ES^p}=\left\lbrace X\in L^\infty:\int_0^p VaR^s(X)ds\leq0\right\rbrace $ and  $\mathcal{Q}_{ES^p}=\left\lbrace \mathbb{Q}\in\mathcal{Q} : \frac{d\mathbb{Q}}{d\mathbb{P}}\leq\frac{1}{p} \right\rbrace$. By taking the distortion $g_x(y)=y(1+x)\wedge1$ we obtain $\rho_x=ES^\frac{1}{1+x}$. In this situation we have \[\mathcal{A}_x=\mathcal{A}_{ES^\frac{1}{1+x}}=\left\lbrace X\in L^\infty\colon ES^{\frac{1}{1+x}}(X)\leq 0\right\rbrace =\left\lbrace X\in L^\infty\colon \int_0^{\frac{1}{1+x}}VaR^s(X)ds\leq 0\right\rbrace.\]

Other properties of Definition \ref{def:accept} are also preserved in the chain of implications in Theorem \ref{Thm:main}. We now focus on Law Invariance and consistency (Monotonicity) concerning SSD. If the space is atomless, Law invariance is inherited by Theorem 2.1  in \cite{Mao2020}.

\begin{Prp}\label{rmk:LI}
	In the conditions of Theorem \ref{Thm:main} we have:
	\begin{enumerate}
		\item  $\alpha$ is law invariant if and only if $\{\mathcal{A}_x\}_{x>0}$ is law invariant ($X\in\mathcal{A}_x$ and $Y\sim X$ implies $Y\in\mathcal{A}_x$) if and only if $\{\rho_x\}_{x>0}$ is law invariant if $\{\alpha_i\}_{i\in\mathcal{I}}$ are law invariant.
		\item $\alpha$ is consistent to SSD if and only if $\{\mathcal{A}_x\}_{x>0}$ is consistent to SSD ($X\in\mathcal{A}_x$ and $Y\succeq X$ implies $Y\in\mathcal{A}_x$) if and only if $\{\rho_x\}_{x>0}$ is consistent to SSD if and only if $\{\alpha_i\}_{i\in\mathcal{I}}$ are consistent to SSD.
	\end{enumerate}
\end{Prp}

\begin{proof}
	For (i), the equivalences arise by following similar steps as those in the proof of (iii)$\implies$(i) of Theorem \ref{Thm:main} by choosing $\mathcal{A}_{Y,x}^\prime =\{X\in L^\infty\colon X\sim Y,\:Y\in \mathcal{A}_{Y,x}\}$. 
	
	Regarding (ii), the steps are similar to those regarding usual Monotonicity in a.s. partial order in Theorem \ref{Thm:main} by taking $\mathcal{A}_{Y,x}^\prime =\{X\in L^\infty\colon X\succeq Y,\:Y\in \mathcal{A}_{Y,x}\}$.
\end{proof}

\begin{Rmk}
	It is possible that $\alpha$ be law invariant but  $\{\alpha_i\}_{i\in\mathcal{I}}$ are not all law invariant. The intuitive reason is that
	law-invariant quasi-concave acceptability indexes respect SSD order, so
	by taking a maximum, one arrives at an acceptability that is
	consistent with SSD order. Not all law-invariant star-shaped acceptability indexes respect SSD order, as for the VaR-based one. 
\end{Rmk}

Under this discussion, we can have the following result for law invariant acceptability indexes in atomless probability spaces directly influenced by distortion-based acceptability indexes, particularly VaR-based ones.

\begin{Prp}\label{crl:LI}
	Under the equivalent conditions in Theorem \ref{Thm:main} on an atomless probability space, we have that:
	\begin{enumerate}
		\item   $\alpha$ is law invariant if and only if there is an increasing family of star-shaped sets $\{G_x\}_{x>0}$ of decreasing functions  $g\colon(0,1)\to\mathbb{R}$ with $g(1^-)\geq 0$ such that	\begin{equation}\label{eq:dual8}
			\alpha(X)=\sup\left\lbrace  x>0\colon\inf_{g\in G_x}\sup\limits_{y>0}\left\lbrace VaR^{\frac{1}{1+y}}(X)-g\left(\frac{1}{1+y}\right) \right\rbrace \leq 0\right\rbrace,\:\forall\:X\in L^\infty. 
		\end{equation}\label{eq:dual7}
		\item   $\alpha$ is SSD consistent if and only if  there is an increasing family of star-shaped sets $\{G_x\}_{x>0}$ of decreasing functions  $g\colon(0,1)\to\mathbb{R}$ with $g(1^-)\geq 0$ such that
		\begin{equation}
			\alpha(X)=\sup\left\lbrace  x>0\colon\inf_{g\in G_x}\sup\limits_{y>0}\left\lbrace ES^{\frac{1}{1+y}}(X)-g\left(\frac{1}{1+y}\right) \right\rbrace \leq 0\right\rbrace,\:\forall\:X\in L^\infty. 
		\end{equation}
	\end{enumerate} 
\end{Prp}

\begin{proof}
	For (i), by Proposition \ref{rmk:LI}, we have that $\alpha$ is law invariant if and only if the family $\{\rho_x\}_{x>0}$ is composed of law invariant star-shaped risk measures. Theorem 12 in \cite{Castagnoli2021} states that
	$\rho\colon L^\infty\to\mathbb{R}$ is a law invariant star-shaped risk measure if and only if there is a star-shaped set $G$ of decreasing functions  $g\colon(0,1)\to\mathbb{R}$ with $g(1^-)\geq 0$ such that
	\[	\rho(X)=\inf_{g\in G}\sup\limits_{\alpha\in(0,1)}\left\lbrace VaR^\alpha(X)-g(\alpha) \right\rbrace,\:\forall\:X\in L^\infty.\] Thus, we take $G_x$ as the set of functions $g$ that represents the star-shaped risk measure $\rho_x$. Note that since $\{\rho_x\}_{x>0}$ is increasing, the family $\{G_x\}_{x>0}$ must be decreasing.
	
	Regarding (ii), by Proposition \ref{rmk:LI} we have that $\alpha$ is SSD consistent if and only if the family $\{\rho_x\}_{x>0}$ is composed by star-shaped SSD consistent risk measures. Moreover, Theorem 11 in \cite{Castagnoli2021} states that $\rho$ is SSD consistent star-shaped risk measure if and only if  there is a star-shaped set $G$ of decreasing functions  $g\colon(0,1)\to\mathbb{R}$ with $g(1^-)\geq 0$ such that
	\[\rho(X)=\inf_{g\in G}\sup\limits_{\alpha\in(0,1)}\left\lbrace ES^\alpha(X)-g(\alpha) \right\rbrace ,\:\forall\:X\in L^\infty.\] Again, we take $G_x$ as the set of maps that represents the star-shaped risk measure $\rho_x$.  
	
\end{proof}

We end this section with an example of a star-shaped risk measure that is neither positively homogeneous nor convex, the Benchmark loss VaR (LVaR) introduced in \cite{Bignozzi2020}. It is defined as $LVaR_\theta(X)=\sup_{t\in\mathbb{R}_+}\{VaR^{\theta(t)}(X)-t\}$, where $\theta\colon\mathbb{R}_+\to[0,1]$ is increasing and right-continuous. In this case we have $\mathcal{A}_{LVaR_\theta}=\left\lbrace X\in L^\infty:\mathbb{P}(X<-t)\leq\theta(t)\:\forall\:t\in\mathbb{R}_+\right\rbrace$. By making $\theta=\alpha$ one recovers the $VaR^\alpha$. By considering an increasing family of maps $\{\theta_x\colon\mathbb{R}_+\to[0,1]\}_{x>0}$ that are increasing and right-continuous we are able to define $\rho_x=LVaR_{\theta_x}$ and generate a star-shaped acceptability index.

\begin{Def}
	The LVaR based acceptability index is a functional $\alpha_{LVaR}\colon L^\infty\to[0,\infty]$ defined as\begin{equation}
		\alpha_{LVaR}(X)=\sup\left\lbrace x>0\colon LVaR^{\theta_x}(X)\leq 0\right\rbrace. 
	\end{equation}
\end{Def}

We get that the acceptance sets in the representation \eqref{eq:dual} are as \[\mathcal{A}_x=\mathcal{A}_{LVaR_{\theta_x}}=\left\lbrace X\in L^\infty:\mathbb{P}(X<-t)\leq\theta_x(t)\:\forall\:t\in\mathbb{R}_+\right\rbrace.\]

\begin{Rmk}\label{rmk:lambda}
	From the idea of LVaR, it might be interesting to discuss
	another similar form of variation of VaR. One of those forms is the $\Lambda VaR$, introduced in \cite{Frittelli2014}. This map is defined for any $X$ as
	$\Lambda VaR(X) = -\inf\{x\in\mathbb{R}: P (X \leq x) \geq\Lambda(x)\}$, 
	for some function $\Lambda\colon\mathbb{R}\to[0,1]$ that is not constantly 0. By making $\Lambda=\alpha$, we recover the traditional $VaR$. Despite being very interesting, $\Lambda VaR$ is not star-shaped. Nonetheless, it is  quasi-star-shaped in the sense that $\Lambda VaR(\lambda X + (1-\lambda)t) \leq \max\{\Lambda VaR(X), \Lambda VaR(t)\},\:\forall\: X\in L^\infty,\:\forall\: t\in\mathbb{R},\:\forall \lambda\in[0,1]$. It will be an interesting extension to study the acceptability indexes based on  $\Lambda VaR$, such as 	$X\mapsto\alpha_{\Lambda VaR}(X)=\sup\left\lbrace x>0\colon \Lambda^x VaR(X)\leq 0\right\rbrace$, where $\{\Lambda^x\}$ is an increasing family of maps. One can even think of other more general classes of  Quasi-star-shaped risk measures. Due to parsimony, we left this pursuit for future study.
\end{Rmk}

\section{Ratio based acceptability indexes}\label{sec:ratio}

A very relevant kind of acceptability index is based on performance measures, typically ratios between some gain or return and a risk measure. In this section, we expose and study the representations of some concrete cases of acceptability indexes based on ratios.

\subsection{Risk adjusted reward on capital}\label{sec:raroc}

A widespread performance measure is RAROC, which is a ratio between the return, measured as expectation, and a risk measure as

\begin{equation*}
	RAROC(X)=\begin{cases}
		\dfrac{E[X]}{\rho(X)}&\text{if} \:E[X]>0\:\text{and}\:\rho(X)>0,\\
		0&\text{if} \:E[X]\leq0\:\text{and}\:\rho(X)>0,\\
		\infty&\text{if}\:\rho(X)\leq 0.
	\end{cases}
\end{equation*} 
RAROC is a coherent or quasi-concave acceptability index if $\rho$ is coherent or convex. This quantity is very useful for performance and even regulation. We now generalize it to our framework by considering a reward measure $\mu$, defined as some risk measure's negative. We keep the nomenclature of risk measures for reward measures $\mu$. The expectation $E[X]$ is a reward measure. The reasoning for $\mu$ is the possibility to consider a more conservative gain in the numerator, with $\mu(X)\leq E[X]$ for any $X\in L^\infty$, for instance. In this case, we have $\alpha_{\mu,\rho}\leq RAROC$. Thus, a possible choice is a quantile to-quantile ratio with  $\mu(X)=-VaR^{p_1}(X)$ and $\rho(X)=VaR^{p_2}(X)$ with $0\leq p_2\leq p_1\leq 1$.

\begin{Def}\label{reward-risk}
	Let $\mu,\rho\colon L^\infty\to\mathbb{R}$ be reward and risk measures, respectively. Then the acceptability index they generate, called risk-adjusted reward on capital, is a functional $\alpha_{\mu,\rho}\colon L^\infty\to[0,\infty]$ defined as 
	\begin{equation}
		\alpha_{\mu,\rho}(X)=\begin{cases}
			\dfrac{\mu(X)}{\rho(X)}&\text{if} \:\mu(X)>0\:\text{and}\:\rho(X)>0,\\
			0&\text{if} \:\mu(X)\leq0\:\text{and}\:\rho(X)>0,\\
			\infty&\text{if}\:\rho(X)\leq 0.
		\end{cases}
	\end{equation} 
\end{Def}

We now prove a result that assures our proposed RAROC is a star-shaped acceptability index and elucidates its representations. All considered properties and results for risk measures exposed so far are valid for reward measures under the sign correction. Thus, we use them without further mention. It is clear that	$\alpha_{\mu,\rho}$ inherits Scale and Law invariance and SSD consistency from $\mu$ and $\rho$.

\begin{Prp}
	Let $\mu,\rho\colon L^\infty\to\mathbb{R}$ be, respectively, Fatou continuous star-shaped reward and risk measures. Then $\alpha_{\mu,\rho}$ is a star-shaped acceptability index. Moreover, $\rho\geq-\mu$ is necessary and sufficient for it can be represented under \eqref{eq:dual3} and \eqref{eq:dual} by \begin{align}
		\rho_x(X)&=-\frac{1}{1+x}\mu(X)+\frac{x}{1+x}\rho(X),\:\forall\:X\in L^\infty,\\
		\mathcal{A}_x&=\left\lbrace X\in L^\infty\colon\dfrac{\mu(X)}{\rho(X)}\geq x \right\rbrace,
	\end{align} 
\end{Prp}

\begin{proof}
	Regarding Monotonicity, for any $X\geq Y$ we have $\mu(X)\geq\mu(Y)$ and $\rho(X)\leq \rho(Y)$. The result is trivially obtained if $\rho(X)\leq 0$. For the case $\rho(X)>0$, if $\mu(X)\leq 0$, the result is also trivial. Otherwise, we thus get \[\alpha_{\mu,\rho}(X)=\frac{\mu(X)}{\rho(X)}\geq\frac{\mu(Y)}{\rho(Y)}=\alpha_{\mu,\rho}(Y).\]  Regarding to Star-shapedness, for any $\lambda\geq 1$ and any $X\in L^\infty$ we have  that $\mu(\lambda X)\leq \lambda\mu(X)$ and $\rho(\lambda X)\geq\lambda\rho(X)$. The result is trivial if $\rho(X)\leq 0$, which is the case for $X=0$. If $\rho(X)>0$ we then get \[\alpha_{\mu,\rho}(\lambda X)=\frac{\mu(\lambda X)}{\rho(\lambda X)}\leq \frac{\lambda\mu(X)}{\lambda\rho(X)}=\alpha_{\mu,\rho}(X).\]
	Weak expectation consistency is obtained directly from the definition. For Fatou continuity, let $\{X_n\}\subseteq L^\infty$ be bounded such that $\lim\limits_{n\to\infty}X_n=X$. Then, both $\mu(X)\geq\limsup\limits_{n\to\infty}\mu(X_n)$ and $\rho(X)\leq\liminf\limits_{n\to\infty}\rho(X_n)$. We then obtain
	\[\alpha_{\mu,\rho}(X)=\frac{\mu(X)}{\rho(X)}\geq\frac{\limsup\limits_{n\to\infty}\mu(X_n)}{\liminf\limits_{n\to\infty}\rho(X_n)}=\lim\limits_{n\to\infty}\left( \frac{\sup\limits_{k\geq n}\mu(X_k)}{\inf\limits_{k\geq n}\rho(X_k)}\right) \geq\limsup\limits_{n\to\infty}\frac{\mu(X_n)}{\rho(X_n)}=\limsup\limits_{n\to\infty}\alpha_{\mu,\rho}(X_n).\] 
	
	Let $\{\rho_x\}_{x>0}$ be a family defined as \[	\rho_x(X)=-\frac{1}{1+x}\mu(X)+\frac{x}{1+x}\rho(X),\:\forall\:X\in L^\infty.\] It is easy to verify that it is a star-shaped risk measure. Further, it increasing in $x$ since we have the following for any $X\in L^\infty$: \begin{align*}
		\frac{\partial\rho_x(X)}{\partial x}&=(1+x)^{-2}(\mu(X)-x\rho(X))+(1+x)^{-1}\rho(X)\\
		&\geq(1+x)^{-2}(-\rho(X)-x\rho(X))+(1+x)^{-1}\rho(X)\\
		&=(1+x)^{-1}\rho(X)-(1+x)^{-1}\rho(X)=0.
	\end{align*}
	
	If $\rho(X)>0$, then for any $x>0$ we have that 
	\begin{align*}
		\alpha_{\mu,\rho}(X)\geq x\iff\mu(X)\geq x\rho(X)	
		\iff&-\frac{1}{1+x}\mu(X)+\frac{x}{1+x}\rho(X) \leq 0\\
		\iff&\rho_x(X)\leq 0.
	\end{align*}
	If $\rho(X)\leq 0$, then $\alpha_{\mu,\rho}(X)=\infty$, which assures $\alpha_{\mu,\rho}(X)\geq x$ for any $x>0$. If $\rho(X)> 0$, then $\mu(X)\geq-\rho(X)\geq0\geq x\rho(X)$. In this case we obtain that \[\rho_x(X)=-\frac{1}{1+x}\mu(X)+\frac{x}{1+x}\rho(X) \leq 0,\:\forall\:x>0.\] Thus, $\alpha_{\mu,\rho}(X)=\infty=\sup\{x>0\colon\rho_x(X)\leq 0\}$. Finally, we get that
	\[\mathcal{A}_x=\mathcal{A}_{\rho_x}=\left\lbrace  X\in L^\infty\colon -\mu(X)+x\rho(X)\leq 0\right\rbrace=\left\lbrace X\in L^\infty\colon\dfrac{\mu(X)}{\rho(X)}\geq x \right\rbrace.  \]
\end{proof}

\subsection{Gain-Loss ratio}\label{sec:GLR}

\cite{Bernardo2000} propose the gain-loss ratio (GLR), which in its usual coherent acceptability form is defined as 
\begin{equation*}
	GLR(X)=\begin{cases}
		\dfrac{E[X]}{E[X^-]}&\:\text{if}\:E[X]>0\:\text{and}\:E[X^-]>0,\\
		0&\:\text{if}\:E[X]\leq 0\:\text{and}\:E[X^-]>0,\\
		\infty&\:\text{if}\:E[X^-]=0.
\end{cases}\end{equation*}
This map is a coherent acceptability index, fulfilling consistencies regarding Expectation and SSD. It can be represented under expectiles. A $p$-expectile of $X$ is the unique solution $y$ of $pE[(X-y)^+]=(1-p)E[(X-y)^-]$ for $p\in[0,1]$. The risk measure expectile Value at risk ($EVaR^p$), which is the negative of p-expectile, is Fatou continuous and coherent for $p\leq\frac{1}{2}$. More specifically, we can take $\rho_x=EVaR^{\frac{1}{2+x}}$ since \[\mathcal{A}_{EVaR^p}=\left\lbrace X\in L^\infty\colon\frac{E[X^+]}{E[X^-]}\geq\dfrac{1-p}{p} \right\rbrace.\] A related acceptability formulation for the GLR is \begin{equation*}
	GLR(X)=\begin{cases}
		\dfrac{E[X^+]}{E[X^-]}&\:\text{if}\:E[X^-]>0,\\
		\infty&\:\text{if}\:E[X^-]=0.
\end{cases}\end{equation*} This formulation is not quasi-concave due to convexity of the numerator. However, it is a star-shaped acceptability index. Hence, it can  be represented under Theorem \ref{Thm:main} by $\rho_x=EVaR^{\frac{1}{1+x}}$, which is not convex for $x<1$, but it is star-shaped for any $x$. 

We propose generalizing this version of GLR. The reasoning for $\mu$ and $\rho$ is similar to RAROC since it brings the possibility of considering a more conservative estimate. For instance, we could to consider with $\mu(X)\leq E[X]$ and $\rho(X)\geq E[-X]$ for any $X\in L^\infty$ and, thus, we would have $\alpha_{\mu,\rho}\leq GLR$.

\begin{Def}\label{GLR}
	Let $\mu,\rho\colon L^\infty\to\mathbb{R}$ be reward and risk measures, respectively. Then the acceptability index they generate, called reward-based gain-loss ratio, is a functional $\alpha_{GLR,\mu,\rho}\colon L^\infty\to[0,\infty]$ defined as 
	\begin{equation}\label{eq:GLR}
		\alpha_{GLR,\mu,\rho}(X)=\begin{cases}
			\dfrac{\mu(X^+)}{\rho(-X^-)}&\text{if} \:\rho(-X^-)>0,\\
			\infty&\text{if} \:\rho(-X^-)=0.
		\end{cases}
	\end{equation} 
\end{Def}

We have to adjust the sign for the loss in $-X^-$ to obtain a positive value due to our pattern for Monotonicity. The quantile-to-quantile ratio is also a possibility here. We now prove properties for this acceptability index.

\begin{Prp}
	Let $\mu,\rho\colon L^\infty\to\mathbb{R}$ be, respectively, Fatou continuous star-shaped reward and risk measures. Then $\alpha_{GLR,\mu,\rho}$ is a star-shaped acceptability index. Moreover, it can be represented under \eqref{eq:dual} by\begin{equation}
		\rho_x(X)=-\sup\left\lbrace y\in\mathbb{R}\colon \mu((X-y)^+)=x\rho(-(X-y)^-)\right\rbrace,\;\forall\:X\in L^\infty.
	\end{equation}
\end{Prp}
\begin{proof}
	
	Monotonicity follows since for any $X\geq Y$, we have both $X^+\geq Y^+$ and $X^-\leq Y^-$. Then $\mu(X^+)\geq\mu(Y^+)$ and $\rho(-X^-)\leq \rho(-Y^-)$. The result is trivially obtained if $\rho(-X^-)= 0$. For the case $\rho(-X^-)>0$, we thus get \[\alpha_{GLR,\mu,\rho}(X)=\frac{\mu(X^+)}{\rho(-X^-)}\geq\frac{\mu(Y^+)}{\rho(-Y^-)}=\alpha_{GLR,\mu,\rho}(Y).\]  Regarding to Star-shapedness, from definition $\alpha_{GLR,\mu,\rho}(0)=\infty$. Further, for any $\lambda\geq 1$ and any $X\in L^\infty$ we have  that $\mu(\lambda X^+)\leq \lambda\mu(X^+)$ and $\rho(-\lambda X^-)\geq\lambda\rho(-X^-)$. If $\rho(-X^-)= 0$, the result is trivial. If $\rho(-X^-)>0$ we then get \[\alpha_{GLR,\mu,\rho}(\lambda X)=\frac{\mu(\lambda X^+)}{\rho(-\lambda X^-)}\leq \frac{\lambda\mu(X^+)}{\lambda\rho(-X^-)}=\alpha_{GLR,\mu,\rho}(X).\]
	Weak expectation consistency is obtained directly from the definition. For Fatou continuity, let $\{X_n\}\subseteq L^\infty$ be bounded such that $\lim\limits_{n\to\infty}X_n=X$. Then  $\lim\limits_{n\to\infty}|X_n|=|X|$. Since $2X^+=|X|+X$ and $2X^-=|X|-X$ we have that both $\lim\limits_{n\to\infty}X^+_n=X^+$ and $\lim\limits_{n\to\infty}X^-_n=X^-$. Thus, $\mu(X^+)\geq\limsup\limits_{n\to\infty}\mu(X^+_n)$ and $\rho(-X^-)\leq\liminf\limits_{n\to\infty}\rho(-X^-_n)$. We then obtain
	\[\alpha_{GLR,\mu,\rho}(X)\geq\lim\limits_{n\to\infty}\left( \frac{\sup\limits_{k\geq n}\mu(X^+_k)}{\inf\limits_{k\geq n}\rho(-X^-_k)}\right) \geq\limsup\limits_{n\to\infty}\frac{\mu(X^+_n)}{\rho(-X^-_n)}=\limsup\limits_{n\to\infty}\alpha_{GLR,\mu,\rho}(X_n).\] 
	
	Furthermore, since $\mathcal{A}_{x}=\left\lbrace X\in L^\infty\colon \alpha_{GLR,\mu,\rho}(X)\geq x \right\rbrace$ is star-shaped acceptance set closed in weak$^{*}$ topology we have that $\rho_{\mathcal{A}_x}$ is a Fatou continuous star-shaped risk measure. By Monotonicity and Translation invariance we have that $g_X\colon\mathbb{R}\to\mathbb{R}$ defined as \[g_X(y)=\mu((X-y)^+)-x\rho(-(X-y)^-)\] is non-increasing, surjective and continuous. Then, $\{y\in\mathbb{R}\colon g_X(y)=0\}\not=\emptyset$ and \[\inf \{y\in\mathbb{R}\colon g_X(y)\geq0\}=\inf\{y\in\mathbb{R}\colon g_X(y)=0\}.\] We get for any $X\in L^\infty$ that \begin{align*}
		\rho_{\mathcal{A}_x}(X)&=\inf\left\lbrace m\in\mathbb{R}\colon\mu((X+m)^+)\geq x\rho(-(X+m)^-)\right\rbrace\\
		&=\inf\{-y\in\mathbb{R}\colon g_X(y)\geq0\}\\
		&=\inf\{-y\in\mathbb{R}\colon g_X(y)=0\}\\
		&=-\sup\{y\in\mathbb{R}\colon g_X(y)=0\}=\rho_x(X).
	\end{align*}
	Since $\mathcal{A}_{\rho_x}=\mathcal{A}_{\rho_{\mathcal{A}_x}}=\mathcal{A}_x$, by Theorem \ref{Thm:main} we get the claim.
\end{proof}

\begin{Rmk}
	We have that $\alpha_{GLR, \mu,\rho}$ inherits Scale and Law invariance and SSD consistency from $\mu$ and $\rho$. When $\rho(X)=-\mu(X)=-E[X]$  the set $\{y\in\mathbb{R}\colon g_X(y)=0\}$ is a singleton and we recover $\rho_x=EVaR^{\frac{1}{1+x}}$. For the general case, there is no guarantee that $\{y\in\mathbb{R}\colon g_X(y)=0\}$ is a singleton. A possible situation is when $\mu$ and $\rho$ are strictly monotone.
\end{Rmk}

\subsection{Reward to deviation ratio}\label{sec:Dev}

Another possibility for a ratio is the one between reward and deviation measures, which we now define.

\begin{Def}\label{def:dev}
	A functional $\mathcal{D}:L^\infty\rightarrow\mathbb{R}_+$ is a deviation measure. It may fulfill the following properties:
	\begin{enumerate}
		\item Non-negativity: For all $X\in L^\infty$, $\mathcal{D}(X)=0$ for constant $X$ and $\mathcal{D}(X)>0$ for
		non-constant X;
		\item Translation insensitivity: $\mathcal{D}(X+C)=\mathcal{D}(X),\:\forall \:X\in L^\infty,\:\forall\:C \in\mathbb{R}$;
		\item Convexity: $\mathcal{D}(\lambda X+(1-\lambda)Y)\leq \lambda \mathcal{D}(X)+(1-\lambda)\mathcal{D}(Y),\:\forall\: X,Y\in L^\infty,\:\forall\:\lambda\in[0,1]$;
		\item Positive homogeneity: $\mathcal{D}(\lambda X)=\lambda \mathcal{D}(X),\:\forall\:X\in L^\infty,\:\forall\:\lambda \geq 0$;
		\item Star-shapedness:  $\mathcal{D}(\lambda X)\geq \lambda\mathcal{D}(X),\:\forall\: X\in L^\infty,\:\forall\:\lambda \geq 1$.
		\item Fatou continuity:  If $\lim_{n\rightarrow\infty}X_n=X$ implies that $\mathcal{D}(X) \leq \liminf\limits_{n\rightarrow\infty} \mathcal{D}( X_{n})$, $\forall\:\{X_n\}_{n=1}^\infty$ bounded in $L^\infty$ norm and for any $X\in L^\infty$.
	\end{enumerate}
	A deviation measure $\mathcal{D}$ is called proper if it fulfills (i) and (ii); convex if it is proper and respects (iii); generalized (also called coherent) if it is convex and fulfills (iv); star-shaped if it is proper and fulfills (v); Fatou continuous if it respects (vi).
\end{Def}

However, a ratio between reward and deviation measures lacks Monotonicity. To see this, let $\mu,\mathcal{D}\colon L^\infty\to\mathbb{R}$ be star-shaped reward and deviation measures, respectively. Let $X>0$ bounded away from zero, $m>0$ and $Y=\lambda X-m$, where $\lambda\geq\left\lVert\frac{X+m}{X}\right\rVert_\infty\geq 1$. Note that $\frac{X+m}{X}\in L^\infty$ since \[\left\lVert\frac{X+m}{X}\right\rVert_\infty\leq 1+m\left\lVert\frac{1}{X}\right\rVert_\infty\leq1+m\frac{1}{\left\lVert X\right\rVert_\infty}<\infty.\] Thus, $Y\geq X$ while the ratio possesses the converse order as \[\dfrac{\mu(Y)}{\mathcal{D}(Y)}=\dfrac{\mu(\lambda X-m)}{\mathcal{D}(\lambda X-m)}<\dfrac{\mu(\lambda X)}{\mathcal{D}(\lambda X)}\leq \dfrac{\mu( X)}{\mathcal{D}(X)}.\] Hence, it does not fit in the framework of acceptability indexes. We then study a monotone version.

\begin{Def}\label{reward-deviation}
	Let $\mu,\mathcal{D}\colon L^\infty\to\mathbb{R}$ be reward and deviation measures, respectively. Then the acceptability index they generate, called monotone reward-deviation ratio, is a functional $\alpha_{\mu,\mathcal{D}}\colon L^\infty\to[0,\infty]$ defined as 
	\begin{equation}
		\alpha_{\mu,\mathcal{D}}(X)=
		\begin{cases}
			\sup\limits_{Y\leq X,\mathcal{D}(Y)>0}\dfrac{\mu(Y)}{\mathcal{D}(Y)}&\text{if} \:\mu(X)>0\:\text{and}\:\mathcal{D}(X)>0,\\
			0&\text{if} \:(\mu(X)\leq 0\:\text{and}\:\mathcal{D}(X)>0)\:\text{or}\: (\mu(X)<0\:\text{and}\:\mathcal{D}(X)=0),\\
			\infty&\text{if}\:\mu(X)\geq 0 \:\text{and}\:\mathcal{D}(X)=0 .		\end{cases}	
	\end{equation} 
\end{Def}

Here, we consider three cases to refine financial reasoning since deviations do not take negative values. We now explore properties and representations for this ratio.

\begin{Prp}
	If $\mu,\mathcal{D}\colon L^\infty\to\mathbb{R}$ are, respectively, Fatou continuous star-shaped reward and deviation measures, then $\alpha_{\mu,\mathcal{D}}$ is a star-shaped acceptability index. Moreover, it can be represented under \eqref{eq:dual3} and \eqref{eq:dual} by \begin{align}
		\rho_x(X)&=\inf\limits_{Y\leq X,\mathcal{D}(Y)>0}\left\lbrace -\mu(Y)+x\mathcal{D}(Y)\right\rbrace,\\ \mathcal{A}_x&=\left\lbrace X\in L^\infty\colon\dfrac{\mu(X)}{\mathcal{D}(X)}\geq x\right\rbrace +L^\infty_+.
	\end{align}
\end{Prp}

\begin{proof}
	Monotonicity follows since for any $X\geq Z$, the cases $\mu(Z)\leq\mu(X)\leq0$ and $\mu(X)\geq0\geq \mu(Z)$ are trivial. Therefore, let $\mu(X)\geq\mu(Z)>0$. Then, we have that $\{Y\in L^\infty\colon Y\leq Z\}\subseteq\{Y\in L^\infty\colon Y\leq X\}$. Hence, \[\alpha_{\mu,\mathcal{D}}(X)=	\sup\limits_{Y\leq X,\mathcal{D}(Y)>0}\dfrac{\mu(Y)}{\mathcal{D}(Y)}\geq\sup\limits_{Y\leq Z,\mathcal{D}(Y)>0}\dfrac{\mu(Y)}{\mathcal{D}(Y)}=\alpha_{\mu,\mathcal{D}}(Z).\]  Regarding to Star-shapedness, by definition $\alpha_{\mu,\mathcal{D}}(0)=\infty$. Further, for any $\lambda\geq 1$ and any $X\in L^\infty$ we have it trivially for the case $\mu(\lambda X)\leq \mu(X)\leq0$. Otherwise, we get that
	\[\alpha_{\mu,\mathcal{D}}(\lambda X)=\sup\limits_{Y\leq \lambda X,\mathcal{D}(Y)>0}\dfrac{\mu(Y)}{\mathcal{D}(Y)}=\sup\limits_{Z\leq X,\mathcal{D}(Z)>0}\dfrac{\mu(\lambda Z)}{\mathcal{D}(\lambda Z)}\leq\sup\limits_{Z\leq X,\mathcal{D}(Z)>0}\dfrac{\mu( Z)}{\mathcal{D}( Z)}=\alpha_{\mu,\mathcal{D}}(X).\]
	Weak expectation consistency is obtained directly from the definition.	For Fatou continuity, let $\{X_n\}\subseteq L^\infty$ be bounded such that $\lim\limits_{n\to\infty}X_n=X$ and $\alpha_{\mu,\mathcal{D}}(X_n)\geq x$ for any $n\in\mathbb{N}$.  Then, we have that for any $\epsilon>0$ there is $\{Y_n^\epsilon\}\subset L^\infty$ such that,  for any $n\in\mathbb{N}$, $Y_n^\epsilon\leq X_n$, $\mathcal{D}(Y_n^\epsilon)>0$ and \[\dfrac{\mu(Y_n^\epsilon)}{\mathcal{D}(Y_n^\epsilon)}>x-\epsilon.\] Note that we can take $\{Y_n^\epsilon\}$ bounded below since $\mu$ is monetary. In fact, take $Y_n^\epsilon\geq K>(x-\epsilon)\mathcal{D}(Y_n^\epsilon)$.
	This implies, jointly to boundedness of $\{X_n\}$, that $\{Y_n^\epsilon\}$ is bounded. Moreover, $Y^\epsilon=\lim\limits_{n\to\infty}Y_n^\epsilon\leq\lim\limits_{n\to\infty}X_n=X$. From Fatou continuity of both $\mu$ and $\mathcal{D}$ we get that \[\alpha_{\mu,\mathcal{D}}(X)\geq\lim\limits_{\epsilon\downarrow0}\dfrac{\mu(Y^\epsilon)}{\mathcal{D}(Y^\epsilon)}\geq\lim\limits_{\epsilon\downarrow0}\limsup\limits_{n\to\infty}\dfrac{\mu(Y_n^\epsilon)}{\mathcal{D}(Y_n^\epsilon)}\geq \lim\limits_{\epsilon\downarrow0}(x-\epsilon)=x.\] 
	
	For the representation under star shaped risk measures, we start showing that the family $\{\rho_{x}\}_{x>0}$, defined as \[\rho_x(X)=\inf\limits_{Y\leq X,\mathcal{D}(Y)>0}\left\lbrace -\mu(Y)+x\mathcal{D}(Y)\right\rbrace\] is composed by star-shaped risk measures and is increasing in $x$. Since the deviations do not take negative values, $x\to\rho_x$ is increasing. For Monotonicity we have that if $X\geq Z$ we have that $\{Y\in L^\infty\colon Y\leq Z\}\subseteq\{Y\in L^\infty\colon Y\leq X\}$. Then $\rho_x(X)\leq \rho_x(Z)$. Translation invariance follows since for any $m\in\mathbb{R}$ we have \begin{align*}
		\rho(X+m)&=\inf\limits_{Y-m\leq  X,\mathcal{D}(Y)>0}\left\lbrace-\mu(Y)+x\mathcal{D}(Y) \right\rbrace\\
		&=\inf\limits_{Z\leq  X,\mathcal{D}(Z)>0}\left\lbrace-\mu(Z+m)+x\mathcal{D}(Z+m) \right\rbrace\\
		&=\inf\limits_{Z\leq  X,\mathcal{D}(Z)>0}\left\lbrace-\mu(Z)+x\mathcal{D}(Z) \right\rbrace+m=\rho_x(X)+m.
	\end{align*}Star-shapedness follows since for any $\lambda\geq 1$ and any $X\in L^\infty$ we have \[\rho_x(\lambda X)=\inf\limits_{Y\leq  X,\mathcal{D}(Y)>0}\left\lbrace-\mu(\lambda Y)+x\mathcal{D}(\lambda Y) \right\rbrace\geq\inf\limits_{Y\leq  X,\mathcal{D}(Y)>0}\left\lbrace\lambda(-\mu(Y)+x\mathcal{D}(Y)) \right\rbrace=\lambda\rho_x(X).\]  For normalization, we have $\rho_x(0)=\rho_x(0^2)\leq 0\rho_x(0)=0$. On the other hand, $\mu(Y)\leq0$ for any $Y\leq 0$. Thus, $-\mu(Y)+x\mathcal{D}(Y)\geq-\mu(Y)\geq 0$ for any $Y\leq 0$. By taking the infimum over those $Y\leq 0$ such that $\mathcal{D}(Y)>0$, we conclude that $\rho_x(0)\geq 0$. Hence $\rho_x(0)=0$. For Fatou continuity, let $\{X_n\}\subseteq L^\infty$ be bounded such that $\lim\limits_{n\to\infty}X_n=X$ and $\rho_x(X_n)\leq y$ for any $n\in\mathbb{N}$. Then, we have that for any $\epsilon>0$ there is $\{Y_n^\epsilon\}\subset L^\infty$ such that,  for any $n\in\mathbb{N}$, $Y_n^\epsilon\leq X_n$, $\mathcal{D}(Y_n^\epsilon)>0$ and $-\mu(Y_n^\epsilon)+x\mathcal{D}(Y_n^\epsilon)<y+\epsilon$. Note that we can take $\{Y_n^\epsilon\}$ bounded below, which implies, jointly to boundedness of $\{X_n\}$, that $\{Y_n\}$ is bounded. Moreover, $Y^\epsilon=\lim\limits_{n\to\infty}Y_n^\epsilon\leq\lim\limits_{n\to\infty}X_n=X$. From Fatou continuity of both $\mu$ and $\mathcal{D}$ we get that \[\rho_x(X)\leq\lim\limits_{\epsilon\downarrow0}\left\lbrace\liminf\limits_{n\to\infty}(\mu(Y_n^\epsilon)+x\mathcal{D}(Y_n^\epsilon)) \right\rbrace \leq \lim\limits_{\epsilon\downarrow0}(y+\epsilon)=y.\] We now show that $\{\rho_x\}_{x>0}$ represents $\alpha_{\mu,\mathcal{D}}$. Since $\mu$ is monotone, the supremum in the definition of $\alpha_{\mu,\mathcal{D}}$ and the infimum in $\rho_x$ can be considered on those $Y$ such that $\mathcal{D}(Y)\leq\mathcal{D}(X)$ without harm. For any $X\in L^\infty$ such that one of the conditions $\mathcal{D}(X)>0$, or $\mu(X)<0$ and $\mathcal{D}(X)=0$ is satisfied, we have that $\alpha_{\mu,\mathcal{D}}(X)\geq x$ if and only if for any $\epsilon>0$ there is $Y_\epsilon\leq X$ with $\mathcal{D}(Y_\epsilon)\in(0,\mathcal{D}(X)]$ such that $\mu(Y_\epsilon)> \mathcal{D}(Y_\epsilon)(x-\epsilon)$. Thus, for any $\epsilon>0$ we get that \[\rho_x(X)\leq -\mu(Y_\epsilon)+x\mathcal{D}(Y_\epsilon)< \epsilon \mathcal{D}(Y_\epsilon)\leq \epsilon\mathcal{D}(X).\]
	By taking the limits for $\epsilon\downarrow0$ we get $\rho_x(X)\leq 0$. 
	If $\mu(X)\geq 0$ and $\mathcal{D}(X)=0$, then $X\in\mathbb{R}_+$ and $\alpha_{\mu,\mathcal{D}}(X)=\infty$, which assures $\alpha_{\mu,\rho}(X)\geq x$ for any $x>0$. On the other hand  \[\rho_x(X)\leq -\mu(X)+x\mathcal{D}(X) \leq 0,\:\forall\:x>0.\] Thus, $\alpha_{\mu,\mathcal{D}}(X)=\infty=\sup\{x>0\colon\rho_x(X)\leq 0\}$. 
	
	Finally, for the family of acceptance sets we have that if $X\in\left\lbrace X\in L^\infty\colon\frac{\mu(X)}{\mathcal{D}(X)}\geq x\right\rbrace +L^\infty_+$, then $X=Z+K$, where $\frac{\mu(Z)}{\mathcal{D}(Z)}\geq x$ with and $K\in L^\infty_+$. Thus, $X\geq Z$, which implies \[\sup\limits_{Y\leq  X,\mathcal{D}(Y)>0}\dfrac{\mu(Y)}{\mathcal{D}(Y)}\geq\sup\limits_{Y\leq  Z,\mathcal{D}(Y)>0}\dfrac{\mu(Y)}{\mathcal{D}(Y)}\geq\dfrac{\mu(Z)}{\mathcal{D}(Z)}\geq x.\] Thus, $X\in\mathcal{A}_x$. For the converse inclusion, let $X\in\mathcal{A}_x$. Then for any $n\in\mathbb{N}$ there is $Y_n\leq X$ with $\mathcal{D}(Y_n)>0$ such that \[\dfrac{\mu(Y_n)}{\mathcal{D}(Y_n)}>x-\dfrac{1}{n}.\] Note that we can take $\{Y_n\}$ bounded below, which implies that this sequence is bounded. Thus, $Y=\lim\limits_{n\to\infty}Y_n\leq X$. Then we get \[\dfrac{\mu(Y)}{\mathcal{D}(Y)}\geq\limsup\limits_{n\to\infty}\dfrac{\mu(Y_n)}{\mathcal{D}(Y_n)}\geq x.\] Hence, $X\in\left\lbrace X\in L^\infty\colon\frac{\mu(X)}{\mathcal{D}(X)}\geq x\right\rbrace +L^\infty_+$.
\end{proof}

\section{Illustration}\label{sec:ill}

We now focus on evaluating the performance measures described above for a simple position in a spot market portfolio. Following \cite{Cherny2009}, this is useful because we get an idea about the
numerical magnitudes of measures and the types of values one may expect to
see for them. In this same vein, the illustration allows an understanding of the relative values of
the various measures as they are all computed for the same data series. Moreover, due to this illustration, there is a direct application for pricing under the physical measure because the acceptability indexes we consider here represent star-shaped levels of law invariant acceptance sets. More specifically, if an agent can access only historical data from net
asset values (returns), he/she may compute the price that attains a particular level of
acceptability. From that, it is possible to price a position to attain a level of acceptability comparable to that
observed in the historical data. 

We consider examples of star-shaped acceptability measures as proposed in the previous sections and their most usual quasi-concave or coherent counterparts. More specifically, we consider $VaR$ and $ES$ based acceptability indexes; RAROC as $E[X]/ES^{0.05}(X)$ as well as the star-shaped reward version (RAROC\_SS) with a quantile to quantile ratio $-VaR^{0.50}(X)/VaR^{0.05}(X)$. The choice for $VaR^{0.50}$ is to reflect the median, which is star-shaped but not convex;  GLR in the usual version $E[X^+]/E[-X^-]$ and the star-shaped reward version (GLR\_SS) computed as $-VaR^{0.50}(X^+)/VaR^{0.05}(-X^-)$; reward to deviation ratio (RDR), where we consider both a coherent $E[X]/ES^{0.05}(X-E[X])$ as star-shaped (RDR\_SS) quantile  $-VaR^{0.50}(X)/(VaR^{0.05}(X)-VaR^{0.50}(X))$. The denominators here are an ES deviation and an inter-quantile range. In addition, we consider combinations: the minimum one, linked to the robust acceptability index that preserves quasi-concavity, and the median and maximum ones, which only preserve star-shapedness. 

The sample comprises 93 stocks of the S\&P 100 from January 5, 2010, to October 22, 2021,
totaling 2973 observations. Due to the lack of data availability for the entire period, some stocks were excluded from the sample. We adopted the S\&P 100 composition in October 2021. We use log returns from daily closing prices (adjusted to splits and dividends). The analysis was conducted during the full sample period, but we also split the portfolio
evaluation into sub-periods to consider different market momentum and conditions. We divided the
out-of-sample period into four sub-samples (number of observations in parenthesis): 2010 to 2013 (1005), 2014 to 2016 (756), 2017 to 2019 (754), and 2020
onward (457). The period between 2010 to late 2021 allows us to analyze
different conditions and market events, including the Greek default from 2010 to 2011, oil
price swings in 2014, Brexit voting in 2016, the US-China trade war that started in 2018
and the beginning of the Covid-19 pandemic in 2020. 

Each of the following Tables 1 to 5 refers to some of the studied windows and brings descriptive statistics from the considered acceptability index computed through the 93 assets. Furthermore, we also compute the same descriptive statistics for the mean returns of these same assets. Results indicate some clear patterns that we mention in the following.

\begin{table}[h!]
	\centering
	\caption{Results for the period 1 - 2010 to 2013}
	\begin{tabular}{lrrrrrr}
		\hline
		$\alpha$& Mean & Stdev & Skewness & Kurtosis & Minimum & Maximum \\ 
		\hline
		Returns & 0.00 & 0.02 & -0.00 & 4.92 & -0.09 & 0.09 \\ 
		VaR & 1.11 & 0.07 & 0.52 & 0.80 & 0.95 & 1.35 \\ 
		ES & 0.02 & 0.01 & 0.40 & 0.58 & 0.00 & 0.05 \\ 
		RAROC & 0.02 & 0.01 & 0.07 & 0.13 & 0.00 & 0.05 \\ 
		RAROC\_SS & 0.03 & 0.02 & 0.39 & -0.34 & 0.00 & 0.09 \\ 
		GLR & 1.16 & 0.07 & -0.22 & 0.84 & 0.93 & 1.34 \\ 
		GLR\_SS & 0.03 & 0.02 & 0.47 & -0.15 & 0.00 & 0.08 \\ 
		RDR & 0.02 & 0.01 & 0.01 & 0.11 & 0.00 & 0.05 \\ 
		RDR\_SS & 0.03 & 0.02 & 0.31 & -0.48 & 0.00 & 0.08 \\ 
		Min & 0.02 & 0.01 & 0.22 & -0.20 & 0.00 & 0.04 \\ 
		Median & 0.03 & 0.02 & 0.57 & 0.25 & 0.00 & 0.08 \\ 
		Max & 1.17 & 0.07 & -0.04 & 0.64 & 0.95 & 1.35 \\ 
		\hline
	\end{tabular}
\end{table}

\begin{table}[h!]
	\centering
	\caption{Results for the period 2 - 2014 to 2016}
	\begin{tabular}{lrrrrrr}
		\hline
		$\alpha$& Mean & Stdev & Skewness & Kurtosis & Minimum & Maximum \\ 
		\hline
		Returns & 0.00 & 0.01 & -0.03 & 5.80 & -0.08 & 0.08 \\ 
		VaR & 1.10 & 0.08 & 0.39 & -0.23 & 0.96 & 1.30 \\ 
		ES & 0.01 & 0.01 & 1.35 & 2.72 & 0.00 & 0.04 \\ 
		RAROC & 0.02 & 0.01 & 1.40 & 4.66 & 0.00 & 0.07 \\ 
		RAROC\_SS & 0.03 & 0.02 & 0.51 & -0.28 & 0.00 & 0.09 \\ 
		GLR & 1.11 & 0.08 & 1.52 & 5.83 & 0.97 & 1.49 \\ 
		GLR\_SS & 0.03 & 0.02 & 0.40 & -0.73 & 0.00 & 0.07 \\ 
		RDR & 0.02 & 0.01 & 1.26 & 3.98 & 0.00 & 0.06 \\ 
		RDR\_SS & 0.03 & 0.02 & 0.42 & -0.46 & 0.00 & 0.08 \\ 
		Min & 0.01 & 0.01 & 1.32 & 2.65 & 0.00 & 0.04 \\ 
		Median & 0.03 & 0.02 & 0.51 & -0.25 & 0.00 & 0.08 \\ 
		Max & 1.14 & 0.08 & 1.26 & 3.65 & 0.99 & 1.49 \\ 
		\hline
	\end{tabular}
\end{table}

\begin{table}[h!]
	\centering
	\caption{Results for the period 3 - 2017 to 2019}
	\begin{tabular}{lrrrrrr}
		\hline
		$\alpha$& Mean & Stdev & Skewness & Kurtosis & Minimum & Maximum \\ 
		\hline
		Returns & 0.00 & 0.01 & -0.23 & 6.74 & -0.08 & 0.08 \\ 
		VaR & 1.18 & 0.10 & 0.11 & 0.35 & 0.86 & 1.46 \\ 
		ES & 0.02 & 0.01 & 0.62 & -0.05 & 0.00 & 0.06 \\ 
		RAROC & 0.02 & 0.01 & 0.06 & -0.72 & 0.00 & 0.05 \\ 
		RAROC\_SS & 0.05 & 0.02 & 0.25 & -0.30 & 0.00 & 0.11 \\ 
		GLR & 1.15 & 0.10 & -0.07 & -0.44 & 0.87 & 1.37 \\ 
		GLR\_SS & 0.05 & 0.03 & 0.21 & -0.40 & 0.00 & 0.11 \\ 
		RDR & 0.02 & 0.01 & 0.02 & -0.77 & 0.00 & 0.05 \\ 
		RDR\_SS & 0.04 & 0.02 & 0.14 & -0.38 & 0.00 & 0.10 \\ 
		Min & 0.02 & 0.01 & 0.51 & -0.31 & 0.00 & 0.05 \\ 
		Median & 0.04 & 0.02 & 0.21 & -0.35 & 0.00 & 0.10 \\ 
		Max & 1.20 & 0.10 & -0.05 & 0.63 & 0.87 & 1.46 \\ 
		\hline
	\end{tabular}
\end{table}

\begin{table}[h!]
	\centering
	\caption{Results for the period 4 - 2020 to late 2021}
	\begin{tabular}{lrrrrrr}
		\hline
		$\alpha$& Mean & Stdev & Skewness & Kurtosis & Minimum & Maximum \\ 
		\hline
		Returns & 0.00 & 0.02 & 0.18 & 9.61 & -0.13 & 0.15 \\ 
		VaR & 1.11 & 0.12 & 0.20 & -0.02 & 0.86 & 1.44 \\ 
		ES & 0.01 & 0.01 & 1.35 & 2.03 & 0.00 & 0.05 \\ 
		RAROC & 0.02 & 0.01 & 0.47 & -0.50 & 0.00 & 0.05 \\ 
		RAROC\_SS & 0.03 & 0.02 & 0.58 & -0.41 & 0.00 & 0.10 \\ 
		GLR & 1.14 & 0.09 & 0.32 & -0.48 & 0.92 & 1.36 \\ 
		GLR\_SS & 0.03 & 0.02 & 0.64 & -0.20 & 0.00 & 0.10 \\ 
		RDR & 0.02 & 0.01 & 0.42 & -0.55 & 0.00 & 0.05 \\ 
		RDR\_SS & 0.03 & 0.02 & 0.49 & -0.56 & 0.00 & 0.09 \\ 
		Min & 0.01 & 0.01 & 1.23 & 1.43 & 0.00 & 0.05 \\ 
		Median & 0.03 & 0.02 & 0.70 & -0.21 & 0.00 & 0.09 \\ 
		Max & 1.16 & 0.10 & 0.62 & -0.04 & 0.99 & 1.44 \\ 
		\hline
	\end{tabular}
\end{table}

\begin{table}[h!]
	\centering
	\caption{Results for the whole period - 2010 to late 2021}
	\begin{tabular}{lrrrrrr}
		\hline
		$\alpha$	& Mean & Stdev & Skewness & Kurtosis & Minimum & Maximum \\ 
		\hline
		Returns & 0.00 & 0.02 & 0.07 & 13.21 & -0.15 & 0.16 \\ 
		VaR & 1.12 & 0.05 & 0.13 & -0.15 & 1.00 & 1.25 \\ 
		ES & 0.01 & 0.01 & 0.49 & -0.65 & 0.00 & 0.03 \\ 
		RAROC & 0.02 & 0.01 & 0.16 & -0.92 & 0.01 & 0.03 \\ 
		RAROC\_SS & 0.03 & 0.01 & 0.03 & -0.18 & 0.00 & 0.07 \\ 
		GLR & 1.14 & 0.05 & 0.15 & -0.94 & 1.04 & 1.24 \\ 
		GLR\_SS & 0.03 & 0.01 & 0.08 & -0.09 & 0.00 & 0.06 \\ 
		RDR & 0.02 & 0.01 & 0.14 & -0.92 & 0.01 & 0.03 \\ 
		RDR\_SS & 0.03 & 0.01 & -0.04 & -0.19 & 0.00 & 0.06 \\ 
		Min & 0.01 & 0.01 & 0.40 & -0.48 & 0.00 & 0.03 \\ 
		Median & 0.03 & 0.01 & 0.21 & -0.18 & 0.00 & 0.06 \\ 
		Max & 1.15 & 0.05 & 0.04 & -0.79 & 1.04 & 1.25 \\ 
		\hline
	\end{tabular}
\end{table}

For magnitude,  both $\alpha_{VaR}$ and GLR present the larger values, with a mean lightly above the unity. This result can be explained in the case of VaR because daily log returns have central tendency measures close to zero. Thus, since the median occurs with significance level $p=0.5$, linked to acceptability level $x=1$, one expects this pattern for VaR-based acceptability indexes. Similar reasoning explains the values for GLR since it represents a ratio of positive and negative parts expectations. The deviation from unity follows the typically negative skewness of daily financial returns. 

The remaining performance measures have a value close to zero because the numerators in the ratios that define them are based either on mean or median, which is very close to zero. At the same time, the denominator assumes values that are much larger and linked to the tails, such as extreme quantiles or ES. $\alpha_{ES}$ also exhibits this magnitude despite not being a ratio because the larger value $ES^p(X)$ may assume $-E[X]$, which is already very close to zero. In fact, except for $\alpha_{VaR}$ and GLR, the acceptability indexes have assumed a zero value for at least some stocks in the sample, as indicated by the minimum statistic.

The standard deviations follow the same pattern, with larger values for $\alpha_{VaR}$ and GLR, naturally explained by their also larger magnitudes. Moreover, in a general way, we report a predominance of positive values for skewness and negative values for kurtosis of the acceptability indexes. This pattern indicates a concentration of values around the global mean with more probability of occurrence above the mean. 

Regarding the sub-samples, crisis periods, such as those in Tables 1 and 4, expose a tendency to reduce the mean value and skewness of acceptability. This pattern partially reflects the higher volatility and the more frequent occurrence of losses in this period. In sub-samples related to steady periods, as in Tables 2 and 3, we realize the presence of more acceptability indexes with a positive kurtosis, indicating the larger concentration of values far from the mean. The frequent occurrence of stocks can partially explain this result and the good performance during this period.

Comparing the star-shaped acceptability indexes with their quasi-concave/coherent counterparts makes it possible to verify that they present similar average values. This pattern indicates that star-shaped acceptability indexes can be used for performance evaluation, producing the same evaluations as the quasi-concave/coherent ones but with more theoretical generality.

Nonetheless, the two groups present discrepancies regarding the changes in distinct sub-samples observed in the previous paragraphs. More specifically, such changes are more prominent and have larger discrepancies for the quasi-concave/coherent performance measures concerning the star-shaped ones, reflecting some robustness from the latter type. This result is related to the fact that the functionals used to define the quasi-concave/coherent ones, expectation and ES, are more sensitive to changes in data than quantiles, which we have considered the star-shaped ones.

\bibliography{ref}
\bibliographystyle{elsarticle-harv}
\end{document}